\documentclass[12pt,a4paper]{amsart}
\makeatletter
\renewcommand\normalsize{%
    \@setfontsize\normalsize{11.7}{14pt plus .3pt minus .3pt}%
    \abovedisplayskip 10\p@ \@plus4\p@ \@minus4\p@
    \abovedisplayshortskip 6\p@ \@plus2\p@
    \belowdisplayshortskip 6\p@ \@plus2\p@
    \belowdisplayskip \abovedisplayskip}
\renewcommand\small{%
    \@setfontsize\small{9.5}{12\p@ plus .2\p@ minus .2\p@}%
    \abovedisplayskip 8.5\p@ \@plus4\p@ \@minus1\p@
    \belowdisplayskip \abovedisplayskip
    \abovedisplayshortskip \abovedisplayskip
    \belowdisplayshortskip \abovedisplayskip}
\renewcommand\footnotesize{%
    \@setfontsize\footnotesize{8.5}{9.25\p@ plus .1pt minus .1pt}
    \abovedisplayskip 6\p@ \@plus4\p@ \@minus1\p@
    \belowdisplayskip \abovedisplayskip
    \abovedisplayshortskip \abovedisplayskip
    \belowdisplayshortskip \abovedisplayskip}
\ifdefined\pdfpagewidth
\else
\fi
\calclayout
\makeatother

\usepackage{graphicx}%
\usepackage{multirow}%
\usepackage{amsmath,amssymb,amsfonts}%
\usepackage{amsthm}%
\usepackage{mathrsfs}%
\usepackage[title]{appendix}%
\usepackage{xcolor}%
\usepackage{textcomp}%
\usepackage{manyfoot}%
\usepackage{booktabs}%
\usepackage{algorithm}%
\usepackage{algorithmicx}%
\usepackage{algpseudocode}%
\usepackage{listings}%
\usepackage{url}
\usepackage[font=small]{caption}

\theoremstyle{plain}%
\newtheorem{theorem}{Theorem}
\newtheorem{proposition}[theorem]{Proposition}%
\newtheorem{lemma}[theorem]{Lemma}%
\newtheorem{assertion}[theorem]{Assertion}%

\theoremstyle{remark}%
\newtheorem{example}{Example}%
\newtheorem{remark}{Remark}%

\theoremstyle{definition}%
\newtheorem{definition}{Definition}%

\begin{document}
\title[Model of Protein Interactions]{A Novel Mathematical Model of Protein Interactions from the Perspective of Electron Delocalization}

\author {Naoto Morikawa}

\begin{abstract}
Proteins are the workhorse molecules of the cell and perform their biological functions by binding to other molecules through physical contact. Protein function is then regulated through coupling of bindings on the protein (\textit{allosteric regulation}). Just as the genetic code provides the blueprint for protein synthesis, the coupling is thought to provide the basis for protein communication and interaction. However, it is not yet fully understood how binding of a molecule at one site affects binding of another molecule at another distal site on a protein, even more than $60$ years after its discovery in $1961$.

In this paper, I propose a simple mathematical model of protein interactions, using a “quantized” version of differential geometry, i.e., the \textit{discrete differential geometry of $n$-simplices}. The model is based on the concept of \textit{electron delocalization}, one of the main features of quantum chemistry, Allosteric regulation then follows tautologically from the definition of interactions.

No prior knowledge of conventional discrete differential geometry, protein science, or quantum chemistry is required. I hope this paper will provide a starting point for many mathematicians to study chemistry and molecular biology.
\end{abstract}


\maketitle

\tableofcontents

\section{Introduction}\label{sec1}

Chemistry is a science of contradictions, and for most mathematicians it would take time to grasp the whole picture of chemistry. First, chemistry is about bonding, and chemical phenomena are often described in terms of chemical bonds \cite{PB1,ZS2}. “Bonding is what separates chemistry from physics” \cite{CH3}. Second, as Charles A. Coulson lamented, “a chemical bond is not a real thing. It does not exist. No one has ever seen one. No one ever can. It is a figment of our own imagination” \cite{CA4}. Third, as Roald Hoffmann says, “any rigorous definition of a chemical bond is bound to be impoverishing” \cite{PB1}. According to Robert S. Mulliken, “the more we know and compute, the more concepts disappear” \cite{RS5}.

In this paper, I propose a simple mathematical model of intermolecular interactions of proteins that does not rely on the concept of “chemical bonding,” using the discrete differential geometry of $n$-simplices \cite{NM6}. This paper does not concern an application of existing mathematics to problems in protein science, but rather a proposal of a novel mathematical framework. The usefulness of the proposed model is demonstrated by addressing a problem in protein science\footnote{i.e., elucidation of the mechanism of allosteric regulation.} in protein science. For simplicity, I only consider the case of $2$-simplices, i.e. triangles. The ultimate goal of the research is a mathematical description of protein interactions that elucidates “protein allosteric regulation” (see below) based on the concept of “electron delocalization” (i.e., sharing of electrons over more than two atoms\footnote{“Chemical bonding” corresponds to a sharing of electrons over two atoms.}), one of the main features of quantum chemistry. I hope this paper will open up a new avenue/shortcut for mathematicians to study chemistry and molecular biology.

“Protein allosteric regulation,” first mentioned in 1961 by Jacques Monod and Francois Jacob \cite{MJ7}, is the coupling between two molecular-binding events on the surface of a protein, where binding at one site (the \textit{functional} site) is affected (i.e., inhibited or activated) by binding at another distal site (the \textit{regulatory} site) \cite{Ce8, LN9, WM10, MW11}. Just as the genetic code provides the blueprint for protein synthesis, allostery provides the basis for protein communication and interaction \cite{MC12}. Allosteric regulation is ubiquitous in living systems and may provide an innovative approach to developing more selective drugs with fewer side effects \cite{GT13, CM14}. Despite the importance, it is not yet fully understood how binding at the regulatory site affects binding at the functional site.

No prior knowledge of conventional discrete differential geometry, protein science, or quantum chemistry is required. A basic understanding of chemistry and quantum mechanics should suffice. Rather, I attempted to write this paper not only to propose a novel mathematical model, but also to provide a concise introduction to the field for researchers in mathematical sciences.

Finally, Genocript (\url{http://www.genocript.com}) is the one-man bio-venture started by Naoto Morikawa in $2000$ which is developing software tools for protein structure analysis.

\begin{figure}
\centering
\captionsetup{width=1.0\linewidth}
\includegraphics{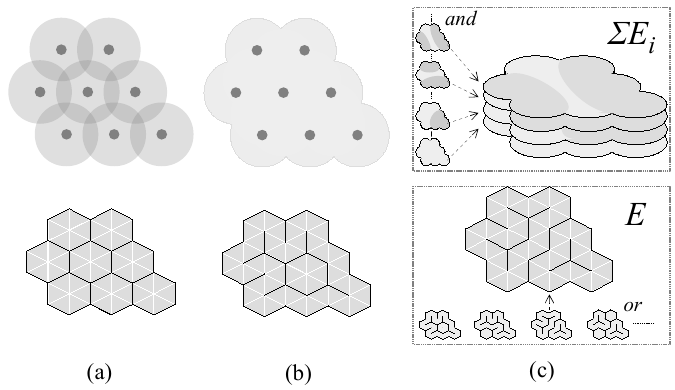}
\caption{The current approach (top) vs. the proposed approach (bottom). 
(a) Interacting atoms. (b) The molecule formed. (c) The state of the molecule.}
\label{figure1}
\end{figure}

\section{Outline of the proposed model}\label{sec2}

\subsection{Fundamentals of quantum mechanics of proteins}\label{sec21}

Proteins are large, complex molecules, with a median size of about $200$ residues \cite{MW15} and an average half-life of about $105$ hours \cite{SE16}. They perform their biological functions by directly interacting with other molecules through physical contact.

According to quantum mechanics, molecules consist of positively charged nuclei (of atoms) embedded in a cloud of negatively charged electrons (Figure \ref{figure1} (b) top). There are no atoms and no bonds within a molecule, only nuclei and indistinguishable electrons. Because of the indistinguishability, electrons cannot be assigned to specific nuclei. Three types of interactions are at work between nuclei and electrons: (1) the long-range (i.e., power-law decay) attraction between nucleus and electrons, (2) the long-range repulsion between electrons, and (3) the short-range (i.e., exponential-law decay) attraction between nuclei, called \textit{bonding interaction}. (1) and (2) are due to electrostatic interactions between charged particles. (3) is due to quantum theoretical interactions between the corresponding atoms. The formation of a bonding interaction\footnote{The bonding interaction of atoms occurs when the resulting molecule has lower energy than it would have if it had been separated.} is considered to be a consequence of the kinetic energy lowering resulting from \textit{delocalization of electrons} over several nuclei (\cite{SI17, LG18}).

The state (i.e., structure and energy) of a molecule is given by the time-independent Schrödinger equation, a second-order partial differential equation involving the spatial coordinates of all the nuclei and electrons. Antisymmetric solutions to the equation are called wave functions, where the sign of wave functions changes when two electrons are exchanged. Antisymmetry is due to the indistinguishability of electrons, which leads to the Pauli repulsion: two electrons can occupy the same region of space only if they have opposite spin.  In principle, wave functions describe all electrons and nuclei in the molecule. However, it is not possible to solve the Schrödinger equation for many-electron molecules, and approximations are required.

First, since the mass of the nucleus is several thousand times greater than the mass of the electron, we usually assume that the nucleus is in a fixed position relative to the electron (the Born–Oppenheimer approximation). Second, we assume that the motion of an electron is independent of all other electrons. In particular, we write the multi-electron wavefunction as an antisymmetrized product of one-electron functions. These one-electron functions, delocalized throughout the internal space of a molecule, are called molecular orbitals.\footnote{If the molecule is an atom, they are called atomic orbitals.}

Then we can solve the Schrödinger equation iteratively using the average potential generated by the other electrons (the Hartree-Fock method). While the orbital approximation of many-electron molecules is a powerful method, it is important to keep in mind that the electrons are indistinguishable from each other. In particular, the orbitals are fictitious and not physically observable. Moreover, “the decomposition of the total wave function into one-electron orbitals is subject to arbitrary decisions depending on what the author considers reasonable. Whatever type of orbitals are used, they are just a model, which play a crucial role in chemistry” \cite{F19}. That is, “for chemists, interpretation based on orbitals is a very intuitive process. Orbital-based definition of, for example, the charge-transfer term is quite natural and aligns well with experimentalists’ ideas” \cite{PC20}.

Figure \ref{figure1} compares two approaches to describing the state of a molecule: the current one (top) and the proposed one (bottom). 

Figure \ref{figure1} (a) top shows eight interacting atoms. In the current approach, each atom consists of a nucleus (small dark gray disk) and an electron cloud (large light gray disk). Note that the electron clouds overlap. 

Figure \ref{figure1} (b) top shows the molecule formed as the result of the interaction between the eight atoms. In the current approach, atoms are \textit{bonded} together to form a molecule through covalent interactions (i.e., \textit{full} sharing of electrons). As the result, the eight nuclei (of atoms) become embedded within a cloud of electrons spread throughout the entire molecule. The arrangement of the nuclei determines the shape of the molecule, while the state of the electron cloud determines the state of the molecule. Regarding interactions between molecules, molecules are attracted each other to form an intermolecular complex through weak but abundant non-covalent interactions (e.g., \textit{partial} sharing of electrons). For details on intermolecular interactions, see Subsection \ref{sec41}.

Figure \ref{figure1} (c) top shows the state (i.e., structure and energy) of the molecule. In the current approach, the state of a molecule is described by a solution to the Schrödinger equation, currently solved using the orbital approximation. Each orbital is spread throughout the molecule and contains up to two electrons. A molecule with $2n$ electrons is then denoted as a \textit{product} of $n$ or more orbitals.\footnote{Some orbitals may be empty.} The energy of a molecule is obtained as the sum of the energies of the occupied orbitals. In this example, the molecule is a product of four orbitals\footnote{The molecule is assumed to be an eight-electron system, i.e., one electron for eact nucleus.}, and its energy is given by $\sum_{1\leq i \leq 4}E_i$, where $E_i$'s are the energies of the four orbitals. 

\subsection{The proposed model of protein interactions}\label{sec22}

\begin{figure}
\centering
\captionsetup{width=1.0\linewidth}
\includegraphics{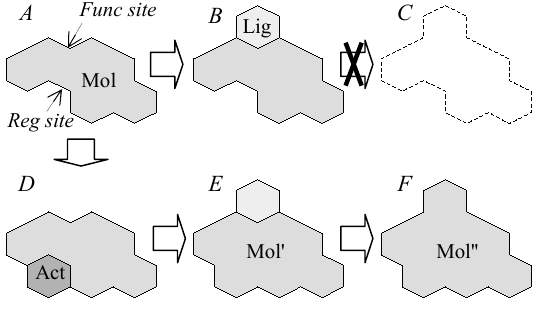}
\caption{Example of allosteric regulation.}
\label{figure2}
\end{figure}

As shown in Figure \ref{figure1} (a) bottom, in the proposed approach, each atom is represented as a loop of six triangles (i.e., a hexagon consisting of six triangles), and the loops do not overlap. 

As shown in Figure \ref{figure1} (b) bottom, in the proposed approach, \textbf{a molecule (such as a protein) is a loop of triangles,\footnote{If there are loops enclosed within a loop, they are considered to be part of the enclosing loop.} where interactions between molecules (such as protein interactions) correspond to fusion and fission of loops}. There is no distinction between covalent interactions and non-covalent interactions. In this example, eight atoms (i.e., loops) are fused together to form a molecule (i.e., loop) of $48$ triangles. In other words, the $48$ triangles are spread throughout the molecule and the six triangles of each atom are now \textit{delocalized} throughout the entire molecule.

As shown in Figure \ref{figure1} (c) bottom, in the proposed approach, \textbf{the state of a molecule corresponds to a division of a loop into a collection of loops (i.e., a \textit{loop complex})}.\footnote{A loop complex is a state of a molecule if the component loops can fuse into a single loop. Otherwise, a loop complex is a molecular complex.} The component loops are not spread throughout the entire molecule, and the state of the molecule is denoted as a \textit{disjoint union} of the component loops. The energy $E$ of the state of a molecule is then given as the sum of the energies of the component loops.

Loop decompositions of a molecule (i.e., divisions of a loop into a loop complex) are computed using flows of triangles (See Section \ref{sec3}). A vector field of triangles is defined on the region occupied by the molecule, and a loop complex is obtained as a collection of closed trajectories contained within the region. The energy $E$ of the state of the molecule is then defined by 
\begin{equation*}
E := \Sigma_{\textit{all loops of the loop complex}} \ 1/ (\textit{loop length}).
\end{equation*}
By definition, the fewer the number of loops contained in a loop complex, the lower the energy of the loop complex.\footnote{It is a consequence of quantum mechanics that a lower energy orbital is formed by constructive interference of two higher energy orbitals.} In particular, a loop passing through the entire region gives a most stable state and is called a \textit{ground state}.\footnote{A region may have more than one loop that passes through the entire region.} In the following, loops passing through the entire region is referred to as \textit{one-stroke} loops for short. The existence of interactions between two molecules is then determined by the existence of one-stroke loops passing through the entire region occupied by the two molecules. 

In Figure \ref{figure1} (c) bottom, four loop decompositions of a molecule are shown (small figures): a one-stroke loop of length $48$, a complex of two loops (length $42$ and length $6$), a complex of two loops (length $30$ and length $18$), and a complex of four loops (two loops of length 18 and two loops of length $6$). Their energies are $E=1/48 \approx 0.02$, $E=1/30+1/18 \approx  0.09$, and $E=2/6+2/18 \approx 0.44$, respectively. Since there is a  one-stroke loop passing through the entire region occupied by the eight atoms, they interact to form a molecule. The selected state (large figure) corresponds to the complex of two loops (length $30$ and length $18$), whose energy is about $0.09$.

\begin{figure}
\centering
\captionsetup{width=1.0\linewidth}
\includegraphics{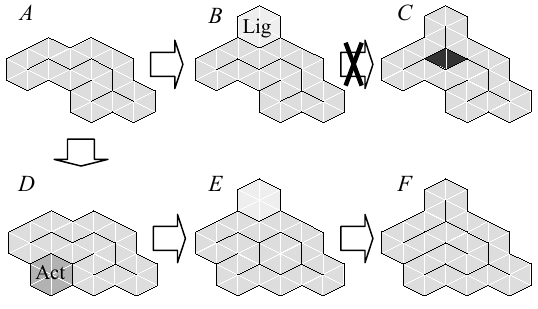}
\caption{The proposed mechanism of allosteric regulation}
\label{figure3}
\end{figure}

Finally, let's briefly explain how the underlying mechanism of allosteric regulation is realized in the proposed model. As explained below, allostery follows directly (tautologically) from the definition of intermolecular interactions.

Figure \ref{figure2} shows a molecule \textit{Mol} with a functional site and a regulatory site (A). The functional site provides a complementary shape to a small molecule \textit{Lig} called \textit{ligand} (B). The regulatory site provides a complementary shape to a small molecule \textit{Act} called \textit{activator} (D). We then suppose that (1) \textit{Act} always binds to the regulatory site, and (2) \textit{Act} regulates the binding of \textit{Lig} to \textit{Mol}, i.e., \textit{Lig} binds to the functional site only when \textit{Act} has already bound to the regulatory site. 

Figure \ref{figure3} shows the proposed mechanism of Figure 2. First, molecules \textit{Mol}, \textit{Lig}, and \textit{Act} are denoted as loops of triangles (A, B and D). We denote the region occupied by \textit{Mol}, \textit{Lig}, and \textit{Act} as $|Mol|$, $|Lig|$, and $|Act|$, respectively. Since there is no one-stroke loop passing through the union of $|Mol|$ and $|Lig|$ (C), they do not interact. On the other hand, since there is a one-stroke loop passing through the union of $|Mol|$ and $|Act|$, they interact to form \textit{Mol’} (E). In the same way, \textit{Mol’} and \textit{Lig} interact to form \textit{Mol’’} (F).

\subsection{Previous studies on discrete differential geometry}\label{sec23a}

Discrete differential geometry (DDG) is often explained in the context of computer science \cite{CW20a, GDS20b}. In computer science, smooth objects (such as smooth curves) are represented as discrete objects (such as polygonal lines). Since the local shape of a discrete object is \textbf{not differentiable}, it is described using vertex positions, edge angles, and similar attributes rather than derivatives. DDG provides the mathematical foundation and algorithms for handling discrete geometric data, which is ubiquitous in modern computing.

Currently, various discretization methods such as the finite difference method (FDM), the finite element method (FEM), and the finite volume method (FVM) are used in numerical simulations of partial differential equations (PDEs). The difference between DDG and those numerical methods for PDEs lies in the fact that discrete differential geometry aims not at discretization of objects/equations, but at discretization of the whole theory of classical differential geometry, where the latter appears as a limit of refinement of the discretization. According to \cite{BS20c}, “there is a common belief that the smooth theories can be obtained in a limit from the corresponding discrete ones.” 

On the other hand, the discrete differential geometry of $n$-simplices (DDGNS) focuses primarily on “quantization” rather than “discretization” of classical differential geometry. Just as classical mechanics does not appear as a smooth limit of quantum mechanics, classical differential geometry does not appear as a smooth limit of DDGNS. The research subject of DDGNS is polygonal lines\footnote{An $n$-dimensional object obtained by connecting $n$-dimensional $n$-simplices one by one via their common faces is referred to as a \textit{polygonal line} by abuse of terminology. In the text, it is called a \textit{trajectory} of $n$-simplices.} consisting of $n$-simplices in \textbf{$R^n$}.

Unlike the polygonal lines studied in DDG, the polygonal lines in DDGNS \textbf{have a second derivative}. Even a protein structure alignment software (\textit{ComSubstruct}) has been developed using second derivatives \cite{NM20d}. Furthermore, one can define morphisms between flows of $n$-simplices in \textbf{$R^n$} using the concept of fusion and fission of loops. This has led to an attempt to use category theory to algebraically describe the shape of closed polygonal lines consisting of $n$-simplices \cite{NM20e}.

\subsection{Previous studies on protein allosteric regulation}\label{sec23}

Proteins exist in a dynamic equilibrium, constantly fluctuating between multiple similar conformations rather than having a single rigid structure. In previous studies, allosteric regulation is often described thermodynamically without detailed insight into its mechanisms. “Phenomenological models have been developed that successfully describe the thermodynamics of allostery but do not reveal its underlying molecular mechanism” \cite{LM21}. In particular, I know of no geometric mechanism that could explain \textbf{allosteric regulation without significant structural alteration}.\footnote{Allosteric regulation \textit{with} significant structural alteration is induced by steric hindrance.}

It is now acknowledged that allostery is due to the re-distribution of existing conformations induced by molecular binding, such as binding of an activator to its regulatory site \cite{MW11,MC12}. Protein activity is then regulated not only by changes in the average conformation induced by changes in enthalpy (i.e., the amount of thermal energy stored), but also by changes in dynamic fluctuations induced by changes in entropy (i.e., the number of possible conformations).

In this scenario, {allosteric regulation without significant structural alteration} corresponds to a “re-distribution of conformations that does not change the average position of the atoms” \cite{CD22}. This type of allostery is driven primarily by entropy changes, where molecular binding may increase thermal fluctuations (i.e., conformational entropy gain) or decrease thermal fluctuations (i.e., conformational entropy loss) \cite{TK23,W24,HR25}.

On the other hand. in the proposed approach, I consider \textbf{energy lowering due to electron delocalization} to analyze {allosteric regulation without significant structural alteration}. This approach is plausible because the energy lowering caused by {electron delocalization} is the driving force behind {bonding interactions} \cite{SI17, LG18}.

Finally, in previous structural studies, proteins are usually described in a bottom-up fashion, where they are typically represented as a graph, with vertices corresponding to amino acids and edges corresponding to chemical bonds. Network analysis is then performed to analyze the mechanism of allosteric regulation \cite{GZ26, NV27, VC28, MW29, CT30}. 

On the other hand, in the proposed approach, \textbf{proteins are described in a top-down fashion}, dealing directly with the electron clouds of entire proteins. In doing so, we can implicitly incorporate an effect of quantum mechanics, i.e., electron delocalization, into the model. Moreover, we can apply global theories of mathematics to the problems of proteins. For example, in this paper, we often encounter a situation where the region covered by a loop has holes. In many cases, we can determine the existence of loop decompositions of the holes without filling them. That is, if a loop decomposition exists for every hole, the loop and the holes constitute a loop decomposition of the entire region.

\section{The discrete differential geometry of triangles}\label{sec3}

This section briefly introduces the discrete differential geometry of triangles. The proposed approach uses this geometry to provide proteins (and molecules in general) with a differential structure. In the following, \textbf{$Z$} denotes the collection of all integers, \textbf{$N$} denotes the collection of all natural numbers, and \textbf{$R^n$} denots the $n$-dimensional Euclidean space. 

\subsection{Flows of triangles}\label{sec31}

\begin{figure}
\centering
\captionsetup{width=1.0\linewidth}
\includegraphics{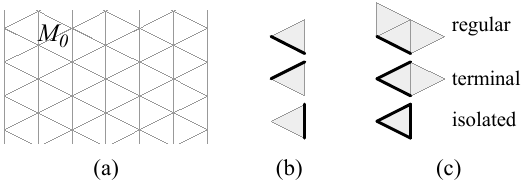}
\caption{(a) The mesh $M_0$. (b) Normal sides (thick lines) of regular triangles. (c) Local flows.}
\label{figure4}
\end{figure}

Figure \ref{figure4} (a) shows the mesh $M_0$ of triangles on which flows of triangles are defined\footnote{Flows on other types of meshes are obtained as flows of triangles induced on the surfaces of a trajectory of $n$-simplices ($n>2$). See Section \ref{sec5}}. $M(M_0)$ denotes the collection of all triangles of $M_0$.

\begin{definition} [Normal vector fields of triangles]\label{NVF_of_T}
A \textit{normal vector field} $V$ on $M_0$ is an assignment of a collection of edges to each triangles of $M_0$, i.e., 
\begin{equation}
V:  M(M_0)\ni t \mapsto V(t) \subset E(t).
\end{equation}
where $E(t)$ is the collection of the three edges of $t$. $NVF(M_0)$ denotes the collection of all normal vector fields on $M_0$.
\end{definition}

\begin{definition} [Normal sides of a triangle]
Given $V \in NVF(M_0)$ and $t \in M(M_0)$. The \textit{normal sides} of $t$ are the edges contained in $V(t)$. Note that $t$ may have more than one normal side.
\end{definition}

In the following, normal sides play the role of \textit{normal vectors}.

\begin{definition}[Classification of triangles]\label{Class_of_T}
Given $V \in NVF(M_0)$ and $t \in M(M_0)$. 
\begin{enumerate}
\item $t$ is called a \textit{regular} triangle (of $V$) if $V(t)$ consists of one edge (Figure \ref{figure4} (b)). $D_1(V)$ denotes the collection of all regular triangles of $V$. 
\item $t$ is called a \textit{branch} triangle (of $V$) if $V(t)$ is empty. $D_0(V)$ denotes the collection of all branch triangles of $V$. 
\item $t$ is called a \textit{terminal} triangle (of $V$) if $V(t)$ conststs of two edges. $D_2(V)$ denotes the collection of all terminal triangles of $V$ by $D_2(V)$. 
\item $t$ is called an \textit{isolated} triangle (of $V$) if $V(t)$ conststs of three edges. $D_3(V)$ denotes the collection of all isolated triangles of $V$.
\end{enumerate}
By definition,
\begin{equation}
M(M_0)=D_0(V) \cup D_1(V) \cup D_2(V) \cup D_3(V).
\end{equation}
$D_1(V)$ is called the \textit{domain} of $V$. $t$ is called a \textit{singular} triangle if it is not regular. 
\end{definition}

\begin{definition}[Regular normal vector fields]
Given $V \in NVF(M_0)$. $V$ is called \textit{b-regular} if $M(M_0)=D_0(V) \cup D_1(V)$. $V$ is called \textit{regular} if $M(M_0)=D_1(V)$.
\end{definition}

\begin{definition}[Local flows of triangles]
Given $V \in NVF(M_0)$ and $s \in M(M_0)$. The \textit{local flow} $F_s$ generated by $V$ at $s$ is the collection of triangles consisting of $s$ and the triangles connected to $s$ by edges other than the normal sides (Figure \ref{figure4} (c)). $F_s$ consists of (1) four triangles if $s \in D_0(V)$, (2) three triangles if $s \in D_1(V)$, (3) two triangles if $s \in D_2(V)$, or (4) one triangle if $s \in D_3(V)$.  
\end{definition}

Connecting local flows, we obtain \textit{integral curves} of the given normal vector field.

\begin{definition}[Trajectories of triangles]
Given $V \in NVF(M_0)$ and $s \in M(M_0)$. The \textit{trajectory} $\psi_s$ of $V$ through $s$ is the unique maximal
chain of triangles obtained by connecting local flows of $V$ starting from $s$ (i.e., the unique maximal \textit{integral curve} of $V$ through $s$). A trajectory $\psi_s$ is called \textit{locally regular} if $\psi_s \subset D_1(V)$, and \textit{regular} if $V$ is regular. 
\end{definition}

A locally regular trajectory $\psi_s$ is a linear sequence of triangles. Suppose that $\psi_s =\{ \cdots, t_0=s, t_1, t_2, \cdots\}$. The local flow $F_{t_n}$ at $t_n \in \psi_s$ is then given by $\{t_{n-1}, t_n, t_{n+1}\}$ ($n \in \mathbf{Z}$).

\begin{definition}[$|\psi|_0$ and $len_0(\psi)$]
Given $V \in NVF(M_0)$ and a trajectory $\psi$ of $V$. $|\psi|_0$ denotes the region swept by $L$, i.e.,
\begin{equation}
|\psi|_0:=\{ t \ |\ t \in \psi \} \subset M(M_0).
\end{equation}
The \textit{length} $len_0(\psi)$ of $\psi$ is the number of triangles contained in $|\psi|_0$, i.e.,
\begin{equation}
len_0(\psi):= \sharp |\psi|_0 \in \textbf{N}.
\end{equation}
\end{definition}

In the proposed model, molecules (including proteins) are denoted as a \textit{loop} of triangles.

\begin{definition}[Loops and loop complices]\label{loops}
Given $V \in NVF(M_0)$ and a locally regular trajectory $\psi$ of $V$. $\psi$ is called a \textit{loop} of $V$ if it is closed. If there are trajectories enclosed within a loop, they are considered to be part of the enclosing loop. \textbf{In other words, loops may have holes inside (i.e., \textit{singularities}).} A \textit{loop complex} is a collection of loops that are in contact.
\end{definition}

\begin{lemma}
Given $V \in NVF(M_0)$ and a loop $L$ of $V$ with enclosed trajectories $\{\psi_1, \psi_2, \ldots, \psi_k\}$. If $L$ is regular, all $\psi_i$ ($1\leq i \leq k$) are loops.
\end{lemma}
\begin{proof}
It follows immediately from the definitions.
\end{proof}

\begin{definition}[$|L|$ and $len(L)$]
Given $V \in NVF(M_0)$ and a loop $L$ of $V$ with enclosed trajectories $\{\psi_1, \psi_2, \ldots, \psi_k\}$. The \textit{extended region} $|L|$ of $L$ is defined by
\begin{equation}
|L|:=\{ t \ |\ t \in L \cup \psi_1 \cup \psi_2 \cup \cdots \cup \psi_k \} \subset M(M_0).
\end{equation}
The \textit{extended length} $len(L)$ of molecule $L$ is defined by
\begin{equation}
len(L):= \sharp |L| \in \textbf{N}.
\end{equation}
Note that there may be holes in $|L|_0$, but not in $|L|$.
\end{definition}

Recall that the state of a molecule is determined by its structure and energy. We define the \textit{energy} of a loop as follows.

\begin{definition}[$E(L)$]\label{E_of_traj}
Given $V \in NVF(M_0)$ and a loop $L$ of $V$. The \textit{energy} $E(L)$ of $L$ is defined by 
\begin{equation}
E(L):= 1 / len(L).
\end{equation}
\end{definition}

\begin{remark}
This is one of the simplest definitions that satisfies the following two constraints (1) $E(L_1)+ E(L_2) > E(L_3)$ if $len(L_1)+ len(L_2) = len(L_3)$\footnote{According to quantum mechanics, bonding orbitals are lower in energy than the individual orbitals.}, and (2) $E(L) \rightarrow 0$ if $len(L) \rightarrow \infty$\footnote{This corresponds to the energy lowering due to electron delocalization.}. 
\end{remark}

Collecting all the trajectories of a given normal vector field, we obtain the flow of the normal vector field.

\begin{definition}[$F_V$]
Given $V \in NVF(M_0)$. The \textit{flow} $F_V$ on $M_0$ generated by $V$ is the disjoint union of trajectories of $V$ that covers the entire $M_0$:
\begin{equation}
F_V:= \sum_{s_i \in A} \psi_{i},
\end{equation}
where $A \subset M(M_0)$ such that (1) $\bigcup_{s_i \in A} {\psi_{i}}= M(M_0)$, and (2) $\psi_i \cap \psi_j = \emptyset$ if $s_i \neq s_j$ ($s_i, s_j \in A$).\footnote{We often write $\psi_i$ instead of $\psi_{s_i}$.} $F_V$ is called \textit{b-regular} if $V$ is b-regular, and \textit{regular} if $V$ is regular.
\end{definition}

\begin{definition}[$E(\Psi)$]\label{E_of_traj2}
Given $V \in NVF(M_0)$ and a loop complex $\Psi=\{L_1, L_2, \cdots, L_k \}$ of $V$. The energy of $\Psi$ is defined by
\begin{equation}
E(\Psi):= \sum_{1 \leq i \leq k} E(L_i).
\end{equation}
\end{definition}

\begin{definition}[$E(F_V)$]
Given $V \in NVF(M_0)$. Suppose that $V$ is regular and the collection of all loops of $V$ is given by $\{L_1, L_2, \cdots, L_k \}$. The energy $E(F_V)$ of $F_V$ is then defined by
\begin{equation}
E(F_V):= \sum_{1 \leq i \leq k} E(L_i).\footnote{Only closed trajectories are counted.}
\end{equation}
\end{definition}

\begin{example} Shown in Figure \ref{figure1} (c) bottom is a loop complex consisting of two loops with lengths $30$ and $18$. The energy $E$ of the loop complex is then given by $E=1/30+1/18 \approx 0.09$.
\end{example}

\subsection{Regular continuations}\label{sec31a}

In this paper, we compute \textit{divisions of a loop $L$ into a loop complex} (Figure \ref{figure1} (c) bottom) using a normal vector field $V$ which is \textit{regular on $|L|$}. Divisions of $L$ are then obtained as a decomposition of $|L|$ into loops of $V$.

First, we construct a normal vector field along the outline $\partial{R}$ of a given region $R$. 

\begin{definition} [$V_P$]
By connecting edges of triangles of $M(M_0)$ one by one, we obtain a polygonal line, say $P$. The \textit{normal vector field $V_P \in NVF(M_0)$ defined by $P$} is given by
\begin{equation}
 V_P (t):=P \cap t \quad (t \in M(M_0)).
\end{equation}
$P$ is called \textit{b-regular} if $V_P$ is b-regular (i.e., $V_P$ has no terminal or isolated triangles).
\end{definition}

We then expand the domain of $V_P$ as follows.
 
\begin{definition} [Regular continuations of $V_P$] Given a polygonal line $P$ and $V_P \in NVF(M_0)$. Assigning normal sides to branch triangles (i.e., triangles with no normal sides) of $V_P$, we obtain $V’ \in NVF(M_0)$ such that (1) $D_1(V') \supsetneq  D_1(V_P)$ and (2) $V’(t)=V_P(t)$ for $t \in D_1(V_P)$. $V’$ is called a \textit{regular continuation} of $V_P$ (over $D_1(V’)$).
\end{definition}

A poligonal line $P$ is usually given as $\partial{R}$ for a region $R$ consisting of finite number of triangles. The normal vector field $V_{\partial{R}}$ is then expanded toward the inside of $R$. When considering loop decompositions of a region, it is sufficient to deal with the following type of regions.

\begin{definition} [Finite boundary regular regions]
Given a region $R$ on $M_0$. $M(R)$ denotes the collection of all triangles contained in $R$. $R$ is called \textit{finite} if $M(R)$ consists of a finite number of triangles. $R$ is called \textit{boundary regular} if $\partial{R}$ is b-regular, where $\partial{R}$ denotes the outline of $R$. 
\end{definition}

\begin{definition} [Regular continuations of $V_{\partial{R}}$]
Given a boundary regular region $R$ on $M_0$ and a regular continuation $V'$ of $V_{\partial{R}}$. $V’$ is called a \textit{regular continuation} of $V_P$ over $R$ if $M(R) \subset D_1(V’)$.
\end{definition}

\begin{figure}
\centering
\captionsetup{width=1.0\linewidth}
\includegraphics{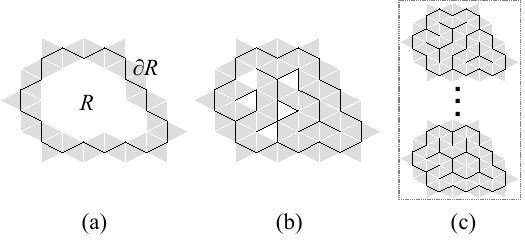}
\caption{(a) Boundary regular region $R$. Regular triangles are colored gray. (b) A regular continuation. It has one branch triangle, three terminal triangles, and one isolated triangle on $R$. (c) Regular continuations of $V_{\partial{R}}$ over $R$.}
\label{figure5}
\end{figure}

\begin{definition} [Locally regular regions]
Given a boundary regular region $R$ on $M_0$. $R$ is called \textit{locally regular} if there is a regular continuation of $V_{\partial{R}}$ over $R$. $R$ is called \textit{regular} if there is a regular continuation of $V_{\partial{R}}$ over the whole space $M(M_0)$. 
\end{definition}

\begin{remark}
In this paper, three types of regularity are defined: (1) regularity of normal vector fields, (2) regularity of trajectories, and (3) regularity of regions.
\end{remark}

Note that we can decompose a region $R$ on $M_0$ into a collection of loops within $R$ if $R$ is locally regular.

\begin{example} 
The region $R$ of Figure \ref{figure5} (a) is boundary regular. In Figure \ref{figure5} (b), $V_{\partial{R}}$ is expanded to a regular continuation of $V_{\partial{R}}$ over a subset of $R$. In Figure \ref{figure5} (c), $V_{\partial{R}}$ is expanded to regular continuations of $V_{\partial{R}}$ over $R$.
\end{example} 

In general, there are more than one regular continuation of $V_{\partial{R}}$ over a given region $R$.

\begin{definition} [$RCNT(R)$]
Given a region $R$ on $M_0$. $RCNT(R)$ denotes the \textit{collection of all regular continuations of $V_{\partial{R}}$ over $R$}.
\end{definition}

\begin{proposition} [Loop decomposition of $R$]
Let $R$ be a finite locally regular region on $M_0$. Suppose that $RCNT(R) \neq \emptyset$. Then, there is a finite collection of loops $\{L_1, L_2, \cdots, L_k\}$ such that (1) $L_i$’s do not intersect each other and (2)
\begin{equation}
R=\bigcup_{1 \leq i \leq k} |L_i|.
\end{equation}
\end{proposition}

\begin{proof}
Let $V' \in RCNT(R)$. The collection of all loops (generated by $V'$) contained within $R$ satisfies the conditions.
\end{proof}

\begin{example} 
In Figure \ref{figure5} (c) top, the finite locally regular region $R$ is decomposed into two non-overlapping loops.
\end{example} 

In quantum mechanics, the \textit{positions of the nuclei} determine the state of a molecule through the \textbf{distribution of electrons} within the electron cloud (Figure \ref{figure1} (c) top). In the proposed model, the \textit{outline of the molecule} determines the state of a molecule through \textbf{loop decomposition} (Figure \ref{figure1} (c) bottom).

\begin{figure}
\centering
\captionsetup{width=1.0\linewidth}
\includegraphics{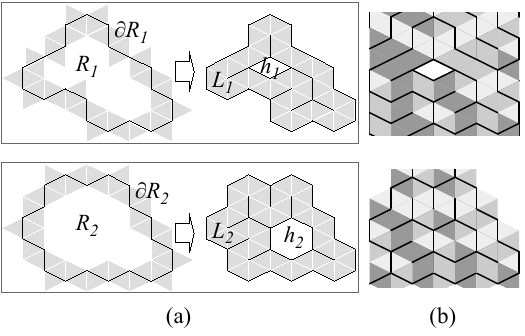}
\caption{(a) Loops with a hole inside obtained by regular continuation. (b) Top view of the corresponding tangent cones.}
\label{figure6}
\end{figure}

When computating loop decompositions by regular continuation, we often encounter a loop with holes inside (Figure \ref{figure6} (a)). In many cases, we can determine if the loop is part of a loop decomposition of the region or not without filling the holes. This topic will be considered in Subsection \ref{sec323}.

\subsection{Computation of loop decompositions} \label{sec32}

Let's begin with an overview of this subsection: First, $M_0$ is placed on the hyperplane $x+y+z=0$ in $\mathbf{R}^3$, and a discrete differential structure is defined on it (Subsections \ref{sec321} and \ref{sec322}). Second, a given region $R$ on $M_0$ is associated with a triangular cone (with multiple tops and no bottom), called \textit{tangent cone} (Subsection \ref{sec323}). Tangent cones are constructed in $\mathbf{R}^3$ by stacking unit cubes diagonally in the direction from $(\infty, \infty, \infty)$ to $(-\infty, -\infty, \infty)$. Finally, a regular continuation $V'$ of $V_{\partial{R}}$ over $R$ is computed using a tangent cone associated with $R$ (Subsection \ref{sec324}). Then, a loop decomposition of $R$ is obtained as the collection of all loops generated by $V'$ contained within $R$.

\subsubsection{Differential structure on $M_0$}\label{sec321}

\begin{figure}
\centering
\captionsetup{width=1.0\linewidth}
\includegraphics{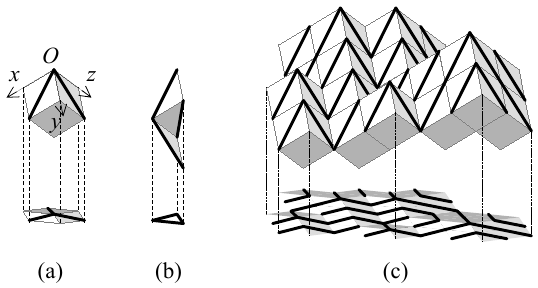}
\caption{(a) Schematic diagram of the unit cube spanned by unit vectors $x$, $y$, and $z$ (top) and its projection on $M_0$ (bottom). Vertical diagonals are indicated by thick lines. (b) Slant triangles over a flat triangle of $M_0$. (c) A tangent cone and the associated normal vector field on $M_0$.}
\label{figure7}
\end{figure}

In the following, point $(l, m, n)$ of $\mathbf{R}^3$ is denoted by $x^ly^mz^n$. We may write ${x_1}^l{x_2}^m{x_3}^n$ instead of $x^ly^mz^n$ when it is convenient (i.e., $x_1=x$, $x_2=y$, and $x_3=z$).

\begin{definition} [$L_3$ and $H$]
The \textit{three-dimensional standard lattice} $L_3$ and the hyperplane $H$ of $\mathbf{R}^3$ are defined by
\begin{align}
L_3& := \{ x^ly^mz^n \in \mathbf{R}^3 \ |\ l, m, n \in \mathbf{Z} \} \subset \mathbf{R}^3,\\
H& := \{ x^ly^mz^n \in \mathbf{R}^3 \ |\ l, m, n \in \mathbf{R}, l+m+n=0 \} \subset \mathbf{R}^3.
\end{align}
The orthogonal projection $\pi$ of $\mathbf{R}^3$ onto $H$ is given by
\begin{equation}
\pi({x^l}{y^m}{z^n}):= x^{(2l-m-n)/3}y^{(-l+2m-n)/3}z^{(-l-m+2n)/3}  \in H.
\end{equation}
\end{definition}

\begin{definition} ($P_1P_2$ and [$P_1, P_2, \cdots, P_k$])
Given $P_1= x^{l_1}y^{m_1}z^{n_1}$ and $P_2 = x^{l_2}y^{m_2}z^{n_2} \in \mathbf{R}^3$. $P_1P_2$ denotes the point $x^{l_1+l_2}y^{m_1+m_2}x^{n_1+n_2}$. Given $P_1$, $P_2$, $\cdots$, $P_k$ and $P \in L_3$. $[P_1, P_2, \cdots, P_k]$ denotes the convex hull defined by $P_1$, $P_2$, $\cdots$, $P_k$, i.e.,
\begin{equation}
[P_1, P_2, \cdots, P_k]:=\left\{ \sum_{1 \leq i \leq k}\lambda_iP_i \ | \ 0 \leq \lambda_i \in \mathbf{R} \  \text{  for  }\  \forall i  \text{, and } \sum_{1 \leq i \leq k}\lambda_i=1  \right\}.
\end{equation}
$P[P_1, P_2, \cdots, P_k]$ is then defined by
\begin{equation}
P[P_1, P_2, \cdots, P_k] := [PP_1, PP_2, \cdots, PP_k].
\end{equation}
\end{definition}

For example, $[P_1, P_2]$ denotes the line segment defined by $P_1$ and $P_2$, and $[P_1, P_2, P_3]$ denotes the triangle defined by three edges $[P_1, P_2]$, $[P_2, P_3]$, and $[P_1, P_3]$. 
 
\begin{example} Given $O=1\ (=x^0y^0z^0)$, $A=x\ (=x^1)$, $B=xy$, $C=y$, $D=yz$, $E=z$, $F=xz$, and $G=xyz \in L_3$ (Figure \ref{figure7} (a)). The unit cube $[O, A, B, C, D, E, F, G]$ has three upper faces $[O, A, B, C]$, $[O, C, D, E]$, $[O, E, F, A]$ and three vertical diagonal lines $[O, B]$, $[O, D]$, $[O, E]$. Then, $[O, A, B, C]$ is divided into two triangles $[O, A, B]$ and $[O, C, B]$ by $[O,B]$.
\end{example}

\begin{definition} [$S$]
Given $p \in L_3$. Triangles $[p, px, pxy]$, $[p, py, pyz]$, $[p, pz, pzx]$, $[p, py, pyx]$, $[p, pz, pzy]$, $[p, px, pxz] \subset \mathbf{R}^3$ are called \textit{slant} triangles. $S$ denotes the collection of all slant triangles, i.e.,
\begin{equation}
S:= \{ [p, px_i, px_ix_j] \subset \mathbf{R}^3 \ |\ p \in L_3, \{i,j\} \subset \{1,2,3\} \text{ such that } i \neq j \}.
\end{equation}
\end{definition}

\begin{definition} [$B$]
Given $p \in L_3$. Triangles $[\pi(p), \pi(px), \pi(pxy)]$, $[\pi(p), \pi(py), \pi(pyz)]$, $[\pi(p), \pi(pz), \pi(pzx)]$, $[\pi(p), \pi(py), \pi(pyx)]$, $[\pi(p), \pi(pz), \pi(pzy)]$, $[\pi(p), \pi(px), \pi(pxz)] \subset H$ are called \textit{flat} triangles. $B$ denotes the collection of all flat triangles, i.e.,
\begin{equation}
B:=\{ [\pi(p), \pi(px_i), \pi(px_ix_j)] \subset H \ |\ p \in L_3, \{i,j\} \subset \{1,2,3\} \text{ such that } i \neq j \}.
\end{equation}
\end{definition}

In this paper, \textbf{we identify $B$ with the collection $M(M_0)$ of all triangles of $M_0$}. A discrete differential structure on $B$ is defined as follows.

\begin{definition} [$S(b)$]
Given $b=[\pi(p), \pi(px_i), \pi(px_ix_j)] \in B$. The \textit{fiber $S(b)$ of $S$ over $b$} is defined by
\begin{align}
&S([\pi(p), \pi(px_i), \pi(px_ix_j)])  \\
&:=\  {\pi_S}^{-1}([\pi(p), \pi(px_i), \pi(px_ix_j)]) \nonumber \\  \nonumber
&=\  \{ \cdots ,\ p[1, x_i, x_ix_j],\ px_i[1, x_j, x_jx_k],\ px_ix_j[1, x_k, x_kx_i],\\  \nonumber
&\phantom{=\  \{ \cdots ,\ }\  px_ix_jx_k[1, x_i, x_ix_j],\ \cdots \}, \nonumber
\end{align}
where $\pi_S$ denotes the projection from $S$ onto $B$ induced by $\pi$, i.e., 
\begin{equation}
\pi_S ([p, px_i, px_ix_j]):= [\pi(p), \pi(px_i), \pi(px_ix_j)] \in B.
\end{equation}
\end{definition}

\begin{definition} [$T$]
The \textit{tangent space} $T$ on $B$ is the quotient of $S$ by an equivalence relation $\sim$, i.e.,
\begin{equation}
T:=S/\sim,
\end{equation}
where the equivalence relation $\sim$ over $S$ is defined by
\begin{equation}
s_1 \sim s_2 \text{ if and only if } \exists n \in \mathbf{Z} \text{ such that } s_1=(x_1x_2x_3)^ns_2 \text{ for } s_1, s_2 \in S.
\end{equation}
\end{definition}

\begin{definition} [$T(b)$]
Given $b=[\pi(p), \pi(px_i), \pi(px_ix_j)] \in B$. The \textit{tangent space $T(b)$ at $b$} (or the \textit{fiber $T(b)$ of $T$ over $b$}) is defined by
\begin{align}
&T([\pi(p), \pi(px_i), \pi(px_ix_j)]) \\
&:=\  {\pi_T}^{-1}([\pi(p), \pi(px_i), \pi(px_ix_j)]) \nonumber \\  \nonumber
&=\  \{p[1, x_i, x_ix_j] \mod \sim,\  px_i[1, x_j, x_jx_k] \mod \sim,\\  \nonumber
&\phantom{=\ \{ }\  px_ix_j[1, x_k, x_kx_i] \mod \sim\}, \nonumber
\end{align}
where $\pi_T$ denotes the projection from $T$ onto $B$ induced by $\pi$, i.e.,
\begin{equation}
\pi_T ([p, px_i, px_ix_j] \mod \sim):= [\pi(p), \pi(px_i), \pi(px_ix_j)] \in B.
\end{equation}
Elements of $T(b)$ are called a \textit{tangent} at $b$.
\end{definition}

\begin{example}
Figure \ref{figure7} (b) shows the tangent space $T(b)$ at $b \in B$.
\end{example}

In this paper, \textbf{we identify the tangents of a flat triangle with the edges of the flat triangle} as shown below.

\begin{definition} [$SS$]
The \textit{side space} $SS$ on $B$ is the collection of all edges of flat triangles, i.e.,
\begin{equation}
SS:=\{ [p_1, p_2], [p_2, p_3], [p_3, p_1] \ |\  [p_1, p_2, p_3] \in B \}.
\end{equation}
\end{definition}

\begin{definition} [$SS(b)$]
Given $b= [p_1, p_2, p_3] \in B$. The \textit{side space $SS(b)$ at $b$} (or the \textit{fiber $SS(b)$ of $SS$ over $b$}) is defined by
\begin{equation}
SS([p_1, p_2, p_3]):= \{[p_1, p_2], [p_2, p_3], [p_3, p_1]\}.
\end{equation}
\end{definition}

\begin{definition} [$NS$]
Given $s=[p, px_i, px_ix_j] \in S$ ($p \in L_3$). The \textit{normal side} $NS(s)$ of $s$ is the edge of $s$ along the slope, i.e.,
\begin{equation}
NS([p, px_i, px_ix_j]):= \pi([p, px_ix_j]) \in SS(\pi_S([p, px_i, px_ix_j])).
\end{equation}
\end{definition}

\begin{definition} [$NS_T$]
Given $t=s \mod \sim \  \in T$. The \textit{normal side $NS_T(t)$ of $t$ induced by $NS$} is  defined by
\begin{equation}
NS_T(t):=NS(s)  \in SS(\pi_S(s)).
\end{equation}
$NS_T$ is well-defined because $NS(s_1)= NS(s_2)$ if $s_1 \sim s_2$.
\end{definition}

\begin{lemma}[$T(b) \leftrightarrow SS(b)$]\label{T_to_SS}
$NS_T$ gives a one-to-one correspondence between the tangent space $T(b)$ and the side space $SS(b)$ for  $\forall b \in B$.
\end{lemma}
\begin{proof}
It follows immediately from the definitions (Figure \ref{figure7} (b)).
\end{proof}

\subsubsection{Normal vector fields on $B$}\label{sec322}

Through the identification of Lemma \ref{T_to_SS}, normal vector fields on $M_0$ will be given using \textit{tangent cones} in Subsection \ref{sec323}. Definition \ref{NVF_of_T} is now rewritten as follows.

\begin{definition} [Normal vector fields on $B$]\label{NVF_on_B}
A \textit{normal vector field} $V$ on $B$ is an assignment of a collection of edges to each flat triangle, i.e.,
\begin{equation}
V:  B \ni b \mapsto V(b) \subset SS(b).
\end{equation}
$X(B)$ denotes the collection of all normal vector fields on $B$. \textbf{We identify $X(B)$ with $NVF(M_0)$.}
\end{definition}

Definition \ref{Class_of_T} is rewritten as follows.

\begin{definition} [Classification of triangles]
Given $V \in X(B)$. Triangles of $B$ are classified into four groups $D_0(V)$, $D_1(V)$, $D_2(V)$, and $D_3(V)$:
\begin{align}
D_0(V):=&\{ b \in B \ |\  \sharp{V(b)}=0 \}, \\
D_1(V):=&\{ b \in B \ |\  \sharp{V(b)}=1 \}, \\
D_2(V):=&\{ b \in B \ |\  \sharp{V(b)}=2 \}, \\
D_3(V):=&\{ b \in B \ |\  \sharp{V(b)}=3 \},
\end{align}
where $\sharp{V(b)}$ denotes the number of edges in $V(b)$. Elements of $D_0(V)$ are called \textit{branch} triangles (of $V$). Elements of $D_1(V)$ are called \textit{regular} triangles (of $V$). Elements of $D_2(V)$ are called \textit{terminal} triangles (of $V$). Elements of $D_3(V)$ are called \textit{isolated} triangles (of $V$). By definition,
\begin{equation}
B=D_0(V) \cup D_1(V) \cup D_2(V) \cup D_3(V).
\end{equation}
$D_1(V)$ is called the \textit{domain} of $V$. Triangles that are not regular are called \textit{singular} triangles (of $V$).
\end{definition}

\begin{definition} [Regular normal vector fields on $B$]
Given $V \in X(B)$. $V$ is called \textit{consistent} if
\begin{equation}
V(b_1) \cap SS(b_2) \subset V(b_2)\  \text{ for } \forall b_1, b_2 \in B,
\end{equation}
i.e., $[a,b] \in V(b_2)$ if $[a,b] \in V(b_1)$ is an edge of a slant triangle over $b_2$. $V$ is called \textit{b-regular} if $B = D_0(V) \cup D_1(V)$ and $V$ is consistent. $V$ is called \textit{regular} if $B = D_1(V)$ and $V$ is consistent. 
\end{definition}

\begin{example}
The normal vector field shown in Figure \ref{figure5} (a) is b-regular but not regular.
\end{example}

\subsubsection{Normal vector fields defined by tangent cones}\label{sec323}

First, let's define tangent cones (with multiple tops and no bottom). 

\begin{definition} [$Cone\  A$]
Given $A \subset L_3$. The tangent cone $Cone\ A$ of $L_3$ generated by $A$ is defined by
\begin{equation}
Cone\ A:=\{ px^ly^mz^n \in L_3 \ |\  p\in A,\  0 \leq l, m, n \in \mathbf{Z}^3 \} \subset L_3.
\end{equation}
(Figure \ref{figure7} (c) top).\footnote{A \textit{cotangent} cone is defined by $Cone^{\ast}\ A:= \{  p(yz)^l(xz)^m(xy)^n \in L_3 \  | \  p\in A,\  0 \leq l, m, n \in \mathbf{Z}^3 \} \subset L_3$.} $TCONE_3$ denotes the collection of all tangent cones of $L_3$. $top(w)$ denotes the collection of all top vertices of $w \in TCONE_3$.
\end{definition}

\begin{definition} [$\partial{w}$]
Given $w \in TCONE_3$ and $p \in L_3$.  
$p$ is called \textit{being on the surfaces} of $w$\footnote{By definiiton, $w$ has three surfaces.} if
\begin{equation}
l_w(p):=\max_{a \in top(w)}\{\min \{l,m,n\ |\  ax^ly^mz^n=p \}\}=0.
\end{equation}
$\partial{w}$ denotes the collection of all slant triangles on the surfaces of $w$, i.e.,
\begin{equation}
\partial{w}:=\{s= [p, px_i, px_ix_j] \in S \ |\  p \in L_3, l_w(p)= l_w(px_i)= l_w(px_ix_j)=0\} \subset S.
\end{equation}
\end{definition}

Each tangent cone specifies a regular normal vector field on $B$. In Subsection \ref{sec324}, we will use tangent cones to compute regular continuations.

\begin{definition} [$\Gamma_T(B)$]
A \textit{section of $\pi_T$ over $B$} is a one-to-one mapping $\phi$ from $B$ to $T$ such that $\pi_T \circ \phi =id_B$. In other words, a section of $\pi_T$ over $B$ is an assignment of a tangent to each flat triangle of $B$. $\Gamma_T(B)$ denotes the collection of all sections of $T$ over $B$.
\end{definition}

\begin{lemma} [$V_{\phi}$]
Given $\phi \in \Gamma_T(B)$. A normal vector field, say $V_{\phi}$, is then defined by
\begin{equation}
V_{\phi}(b):=NS_T(\phi(b)) \subset SS(b) \quad (b \in B).
\end{equation}
$V_{\phi}$ is called the \textbf{normal vector field defined by a section $\phi$ of $\pi_T$ over $B$}. $\phi$ is called \textbf{consistent} if $V_{\phi}$ is consistent. By definition, $V_{\phi}$ is regular if $V_{\phi}$ is consistent.
\end{lemma}

\begin{proof}
It follows immediately from the definitions.
\end{proof}

Projecting the normal sides of $s \in \partial{w}$ of a tangent cone $w$ onto $B$, we obtain a regular normal vector field on $B$ (Figure \ref{figure7} (c)):

\begin{lemma} [$V_w$]
Given $w \in TCONE_3$. A one-to-one mapping $\Gamma_w$ from $B$ to $\partial{w} \subset S$ is uniquely determined by
\begin{equation}
\pi_S \circ \Gamma_w(b)=\pi_S(\Gamma_w(b))=b,
\end{equation}
which induces a section $\phi_w$ of $T$ by
\begin{equation}
\phi_w (b):= \Gamma_w(b) \mod \sim.
\end{equation}
A regular normal vector field $V_w$ on $B$ is then defined by
\begin{equation}
V_w(b):= NS_T (\phi_w(b)) \subset SS(b) \quad  (b \in B).
\end{equation}
$V_w$ is called the \textbf{regular normal vector field defined by a tangent cone $w$}.
\end{lemma}

\begin{proof}
It follows immediately from the definitions.
\end{proof}

\subsubsection{Computation of regular continuations}\label{sec324}

Here are some examples of regular continuation of $V_{\partial{R}}$ over a given region $R$. They are computed using a tangent cone associated with a given region $R$. Note that a loop decomposition of $R$ is obtained as the collection of all loops (generated by the regular continuation) contained within $R$.

\begin{figure}
\centering
\captionsetup{width=1.0\linewidth}
\includegraphics{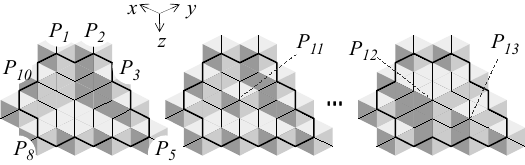}
\caption{Top view of tangent cones $Cone\ A$, $Cone\ A \cup \{P_{11}\}$, and $Cone\ A\cup \{P_{12}, P_{13}\}$ (from left to right), where $A=\{P_1, P_2, \cdots, P_{10}\}$ such that $P_1=1$, $P_2=y/x$, $P_3=yz/x^2$, $P_4=yz^2/x^3$, $P_5= yz^3/x^4$, $P_6= z^3/x^3$, $P_7= z^3/x^2y$, $P_8= z^3/xy^2$, $P_9= z^2/y^2$, $P_{10}=z/y$, $P_{11}=z^2/x$, $P_{12}=z/x^2y$, and $P_{13}=z^2/x^3$. In the figure, the normal sides of slant triangles are represented by black lines. (See also Figure \ref{figure5}.)}
\label{figure8}
\end{figure}

\begin{definition} [Affine region]
Given a finite locally regular region $R \subset B$. $R$ is called \textit{affine} if 
\begin{equation}
\exists w \in TCONE_3 \text{ such that } V_w \in RCNT(R).
\end{equation}
$w$ is called an \textit{tangent cone associated with $R$}.\footnote{A tangent cone is associated with $R$ if and only if the projection image of the vertical diagonals on the tangent cone includes $\partial{R}$.} $ASSOC(R)$ denotes the collection of all tangent cones associated with $R$. By definition, $RCNT(R) \neq \emptyset$ if $ASSOC(R) \neq \emptyset$.
\end{definition}

\begin{example}
$R$ of Figure \ref{figure5} (a) is associated with $Cone\ A$, $Cone\ A \cup \{P_{11}\}$, and $Cone\ A\cup \{P_{12}, P_{13}\}$ of Figure \ref{figure8}.
\end{example}

\begin{remark}
 Associating a region $R$ with a collection of tangent cones locally (i.e., $\partial{R}$ is collectively covered by the collection of tangent cones), we obtain a normal vector field on $R$ that is not necessarily regular.
\end{remark}

\begin{example}
In Figure \ref{figure6} (b) top, $\partial{R_1}$ is collectively covered by three tangent cones. The normal vector field is not regular because of the hole $h_1$.\footnote{However, $RCNT(R_1) \neq \emptyset$ as shown in Figure \ref{figure9} B.} On the other hand, in Figure \ref{figure6} (b) bottom, $R_2$ is affine and the normal vector field is regular.
\end{example}

The two loops in Figure \ref{figure6} (a) have a hole inside; one loop is part of a loop decomposition (bottom), but the other is not (top). In many cases, it is possible to determine if a loop with holes inside is part of a loop decomposition or not without filling the holes with loops.

\begin{example}
In Figure \ref{figure6}, the hole $h_2$ is affine and $RCNT(h_2) \neq \emptyset$. Therefore, $L_2$ is part of a loop decomposition of $R_2$. On the other hand, the hole $h_1$ is not affine. If the following assertion is correct, $RCNT(h_1)=\emptyset$ and $L_1$ is not part of a loop decomposition of $R_1$. 
\end{example}

I have no proof for the following assertion.
\begin{assertion}
Given a (finite) locally regular region $R \subset B$. $R$ is affine if $RCNT(R)\neq \emptyset$.
\end{assertion}

\section{Mathmatical model of protein allosteric regulation}\label{sec4}

Studies of intermolecular interactions have been primarily considering local properties, such as shape complementarity between the protein and the ligand \cite{Ce8}. However, protein allostery implies that the events of intermolecular interactions are not local. Here we propose a \textit{global} model of intermolecular interactions, from which allosteric regulation follows tautologically. For an overview of protein allosteric regulation, see Section \ref{sec1} and Subsection \ref{sec23}.

\subsection{Fundamentals of intermolecular interactions}\label{sec41}

 “Chemistry of the $20$th century was about intramolecular interactions; chemistry of the $21$st century will be about intermolecular interactions” \cite{Be31}. Intermolecular interactions, also called non-covalent interactions, are weak forces between molecules, such as hydrogen bonds, van der Waals' interactions, and Coulombic interactions \cite{AS32,NR33,CM34}. Since non-covalent interactions often work in concert, “even small individual energy contributions may after summation play a significant role” \cite{Cl35}.

Roughly speaking, atoms are \textit{bonded} together to form \textbf{molecules} through covalent interactions (i.e., \textit{full} sharing of electrons), and molecules are attracted each other to form \textbf{intermolecular complexes} through weak but abundant non-covalent interactions. Non-covalent interactions include (1) \textit{partial} sharing of electrons \cite{We36}, (2) electrostatic interactions induced by dynamic electron correlation \cite{PC20,VG37,ZP38,He39}, (3) the entropy-driven hydrophobic effect \cite{PF40}, and (4) others. Hydrogen bonding is a typical example of partial sharing of electrons \cite{WK41,GB42}. London dispersion (a type of the van der Waals interaction) is a typical example of electrostatic interactions induced by dynamic electron correlation \cite{Fe43, HD44, Hu45}. In the case of protein-ligand binding, hydrogen bonding mainly determines binding specificity, while van der Waals and hydrophobic interactions mainly determine binding affinity \cite{MW15}.

\begin{remark}
Electrostatic interactions are simply electromagnetic interactions in which the magnetic component is neglected, and are a good approximation for the purpose of calculating the static properties of atoms and molecules. 
\end{remark}

\begin{remark} Of the four fundamental forces in physics, it is the electromagnetic force that is relevant to interactions between molecules. Therefore, we can say that all non-covalent interactions are quantum mechanical in origin and are driven by electrostatic interactions. The question is which contribution is stronger, electrostatics or quantum mechanics (i.e., the dynamics of the electromagnetic field)? See also Subsection \ref{sec21}.
\end{remark}

The driving force behind intermolecular complex formation is the decrease in Gibbs free energy $G$ of the system, which is defined by $G=H-TS$ using enthalpy $H$, temperature $T$, and entropy $S$ \cite{Ch46, ZR47}. Simply put, Gibbs free energy $G$ is the \textit{chemical potential} stored in the arrangement of atoms within a molecule and available to do work {for free}. The interaction energy $\triangle{G}$ is calculated by subtracting the energy in its initial state from the energy in its final state: $\triangle{G}:=G_{final} - G_{initial}$. On the other hand, the \textit{strength} of interactions corresponds to the difference $\triangle{H}$ of enthalpies: $\triangle{H}:=H_{final} - H_{initial}$. That is, the stronger the bond, the larger the absolute value of $\triangle{H}$.\footnote{Since $G$ and $H$ are an extensive property (like mass or volume), their units are \textit{kcal/mol}.}

\begin{remark}
Enthalpy $H$ is a measure of the total energy (i.e.,  the amount of thermal energy stored), and the change $\triangle{H}$ in $H$ is equal to the energy released ($\triangle{H}<0$) or absorbed ($\triangle{H}>0$) during any process that occurs at constant pressure. 
\end{remark}

\begin{remark}
Entropy $S$ is a measure of uncertainty \cite{Le48}, and its change $\triangle{S}$ indicates the overall decrease ($\triangle{S}<0$) or increase ($\triangle{S}>0$) in the number of microscopic states of a system (i.e., the degree of the freedom of the system). Macromolecules (e.g., proteins) have many degrees of freedom and a variety of conformations, which can lead to a large entropic contribution. For example, the conformational entropy of a protein is defined by the distribution of conformational states populated by the protein. Interactions then cause a redistribution of the populated states.
\end{remark}

\begin{remark}
Tight binding often involves a favorable enthalpy change ($\triangle{H}<0$\footnote{The binding process releases thermal energy, making the products more stable than the reactants.}), an unfavorable entropy change ($\triangle{S}<0$) due to the reduced mobility, and a favorable entropy change ($\triangle{S}>0$) in its surrounding environment. 
\end{remark}

The strength of covalent bonds (i.e., the difference between the total energy of the bonded atoms and the total energy of the separated atoms) typically ranges from $-50$ to $-200$ kcal/mol. For example, the bond energies of a C–C bond, a C=C bond, and a C$\equiv$C bond are about $-85$, $-145$, and $-200$ kcal/mol, respectively \cite{Ro49}. Covalent bond breaking and covalent bond formation are the building blocks of chemical reactions\footnote{Covalent bonds of proteins are not broken during their lifetime.}.

On the other hand, the overall strength of the noncovalent interactions usually ranges from $-0.5$ to $-50$ kcal/mol \cite{AS32}, typically on the order of $-1$ to $-5$ kcal/mol \cite{WK50}. For example, the energy of hydrogen bonds is usually in the range from $-3$ to $-15$ kcal/mol \cite{GB42}. Van der Waals interactions are relatively weak, typically ranging from $-0.5$ to $-1$ kcal/mol \cite{RK51}. In the case of protein-ligand binding, the binding affinity of ligands rarely exceeds $-15$ kcal/mol, which corresponds to an average bound lifetime of less than one day \cite{KC52, Co53}. Since the half-life of most proteins is less than one day, it is believed that there was no evolutionary pressure to create stronger binding \cite{SE16}.

\begin{remark} At room temperature ($25$ ${}^\circ{C}$ or $77$ ${}^\circ{F}$), the average kinetic energy of an ideal (non-interacting) gas is about $-0.9$ kcal/mol (the equipartition theorem \cite{WK54}). On the other hand, the overall stability of a protein (difference in free energy between folded and unfolded states) ranges from $-5$ to $-14$ kcal/mol \cite{MW15}. Therefore, proteins are stable at room temperature.
\end{remark}

\subsection{The proposed model of intermolecular interactions}\label{sec42}

Now, let's give the \textit{global} definition of intermolecular interactions proposed in this paper. In the model, the energy $E(L)$\footnote{See Definition \ref{E_of_traj}.} of a molecule $L$ plays the role of Gibbs free energy. The contribution of entropy is not explicitly considered.

\begin{definition} [Molecules]
A \textit{molecule} is a loop\footnote{See Definition \ref{loops}.} of a consistent normal vector field on $B$. If there are trajectories enclosed within a molecule, they are considered to be part of the enclosing molecule. In other words, molecules may have holes inside (i.e., singularities).
\end{definition}

When considering molecules, we often omit the underlying normal vector field and instead consider the \textit{minimal} vector field necessary to define the molecule.

\begin{definition} [$V_L$]
Given a molecule $L$ of a consistent normal vector field $V$ on $B$.
The \textit{normar vector field $V_L$ associated with $L$} is defined by
\begin{equation}
V_L(b):= 
\begin{cases}
V(b) \subset SS(b)  &\text{ if } b \in L, \\
V(b)\cap \{V(b')\ |\  b' \in L\}  \subset SS(b)   &\text{ if } b \not\in L.
\end{cases}
\end{equation}
By definition, $V_L$ is b-regular and consistent.
\end{definition}

\begin{definition} [$E(L)$]
Let $L$ be a molecule with enclosed loops {$L_1, L_2, \cdots, L_k$}. The \textit{(free) energy $E(L)$ of $L$} is defined by
\begin{equation}
E(L):=1/ len(L).
\end{equation}
\end{definition}

\begin{definition} [Intermolecular complexes]
An \textit{intermolecular complex} is a loop complex\footnote{See Definition \ref{loops}.} of a consistent normal vector field on $B$.
\end{definition}

We often omit the underlying normal vector field and instead consider the \textit{minimal} vector field necessary to define the intermolecular complex.

\begin{definition} [$V_C$]
Given an intermolecular complex $C=\{L_1, L_2,  \cdots , L_k\}$ of a consistent normal vector field $V$ on $B$.
The \textit{normar vector field $V_C$ associated with $C$} is defined by
\begin{equation}
V_C(b):= 
\begin{cases}
V(b) \subset SS(b)  &\text{ if } \exists i \in [1, k] \text{ s.t. }b \in L_i , \\
V(b)\cap \{V(b')\ |\  b' \in \bigcup_{1 \leq i \leq k} L_i \} \subset SS(b)  &\text{ if } b \not\in L_i \text{ for } \forall i \in [1, k].
\end{cases}
\end{equation}
By definition, $V_C$ is b-regular and consistent.
\end{definition}

\begin{definition} [$E(C)$]
Given an intermolecular complex $C=\{L_1, L_2,  \cdots , L_k\}$. 
The \textit{(free) energy $E(C)$ of $C$} is defined by
\begin{equation}
E(C):= \sum_{1 \leq i \leq k} E(L_i).\footnote{See Definition \ref{E_of_traj2}.}
\end{equation}
\end{definition}

\begin{definition} [Intermolecular interactions]\label{M_of_II}
Given two molecules $L_1$ and $L_2$. \textit{$L_1$ and $L_2$ interact} if 
\begin{equation}
\text{ there is a molecule } L_3 \text{ such that } |L_3|=|L_1| \cup |L_2|.
\end{equation}
Note that interactivity depends on a \textit{global} property, i.e., the overall shape of $|L_1|\cup |L_2|$. $L_3$ is denoted by $L_1+L_2$, but may not be determined uniquely\footnote{The region $|L_1|\cup |L_2|$ may have several one-stroke loops that sweep through itself.}. By definition, 
\begin{equation}
E(L_1+L_2) \leq E(L_1)+E(L_2),
\end{equation} 
if molecules $L_1$ and $L_2$ interact.
\end{definition}

\begin{remark} Given three molecules $L_1$, $L_2$, and $L_3$. Then, $(L_1+L_2)+L_3=(L_1+L_3)+L_2$ if both sides of the equation are well defined. Allostery implies one side of the equation is not well defined.
\end{remark}

\begin{figure}
\centering
\captionsetup{width=1.0\linewidth}
\includegraphics{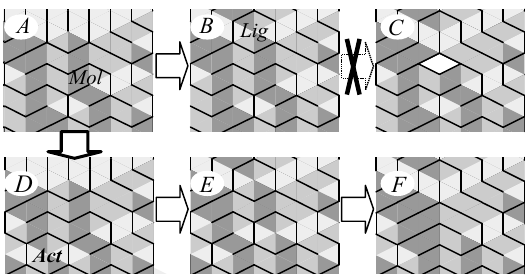}
\caption{Schematic diagram showing the tangent cones that provide the flows of Figure \ref{figure3}.}
\label{figure9}
\end{figure}

\begin{proposition} [Sufficient conditions for reactivity]
Given an intermolecular complex $C=\{L_1, L_2\}$ of two molecules. Suppose that there is a loop of length $k$, say $L_3=\{b_1, b_2, \cdots, b_k\} \subset B$, such that

\begin{equation}
\begin{cases}
&|L_3| \subset|L_1|\cup |L_2| \subset B, \\
&\partial(|L_1|\cup |L_2|) \subset \bigcup_{1 \leq i \leq k} V_{L_3}(b_i) \subset \bigcup_{1 \leq i \leq k} SS(b_i).
\end{cases}
\end{equation}

That is, $L_3$ is contained within $|L_1|\cup |L_2|$, and the normal sides of the normal vector field $V_{L_3}$ associated with $L_3$ contains the boundary of $|L_1|\cup |L_2|$.\footnote{See Figure \ref{figure6} (a) for examples.} Then, $L_1$ and $L_2$ interact if all the holes in $L_3$ are affine.
\end{proposition}

\begin{proof}
Let $\{h_1, h_2, \cdots, h_k\}$ be the holes inside $L_3$ ($h_i$'s are regions on $B$). Since $h_i$’s are affine, $RCNT(h_i)\neq \emptyset$ for all $i$. In particular, $h_i$ has a loop decomposition generated by a regular continuation of $\partial{h_i}$ over $h_i$ for all $i$. Then, $L_3$ and the loop decompositions of the holes $h_i$ ($1 \leq i \leq k$) constitute a loop decomposition of $|L_1|\cup |L_2|$.
\end{proof}

\begin{example} [Allosteric regulation]
Let’s consider the case of Figure \ref{figure2}. Shown in Figure \ref{figure9} are top view of the tangent cones that provide the flows corresponding to Figure \ref{figure2}.\footnote{See also Figure \ref{figure3}}. First, $Lig$ and $Mol$ don’t interact because there are no one-stroke loops for $Lig\cup Mol$. For example, the hole of the loop shown in (C) is not affine. Second, $Act$ and $Mol$ interact because the hole of the loop shown in (E) is affine. Finally, $Lig$ and $Act+Mol$ also interact because the loop shown in (F) has no hole.
\end{example}

\subsection{An example of protein allosteric regulation}\label{sec43}

To get a sense of what intermolecular interactions are, let’s consider \textit{small molecule allosteric drugs} that targets mutated proteins \cite{LK55, XB56, GT57}.
\begin{remark}
\textit{Small molecule drugs} are chemically synthesized compounds with a low molecular mass, typically below $500$ g/mol.
\end{remark}

\subsubsection{Actual examples}\label{sec431}

Diseases are often caused by mutant proteins  \cite{CF58}, i.e., mutations in the DNA sequence of the genomes. In particular, most cancers occur due to mutant proteins. Therefore, disease-causing mutant proteins are the most common targets for therapeutic drugs \cite{HM59}.  Among the most frequently mutated genes in cancer are RAS and TP53, where RAS is the most frequently mutated oncogene (which become oncogenes when mutated) and TP53 is the most frequently mutated tumor suppressor gene (which lose their function when mutated) \cite{DC60}. Here we provide an example of small molecule inhibitors that target the protein products of RAS.

The Kristen-RAS (K-RAS), a product of a RAS gene in humans, is a membrane-bound small protein composed of $188$ amino acids with a molecular mass of $21.6$ kg/mol and a half-life of approximately $24$ hours. K-RAS acts as a molecular switch in intracellular signaling pathways, converting\footnote{hydrolyzing} a molecule called GTP into another molecule called GDP. To transmit signals, K-RAS must be activated by binding to a GTP molecule. Upon GTP binding, K-RAS adopts an active conformation and interacts with a protein called RAF to initiate downstream signaling cascades. K-RAS then returns to its inactive state by converting the bound GTP to GDP. The conversion of GTP to GDP is significantly accelerated when a protein called GAP binds to K-RAS. The release of GDP and the binding of another GTP is facilitated by a protein called GEF. Since K-RAS has a high affinity for GDP and GTP, it is virtually impossible to exchange GDP for GTP without the help of GEF. In this way, K-RAS cycles between two states: GTP-bound (\textit{on}) and GDP-bound (\textit{off}) in normal cells. 

Most K-RAS mutations occur at residue 12 which is normally occupied by glycine \cite{CO61}. Mutation of residue 12 from glycine (G) to cysteine (C) prevents the formation of van der Waals interactions between K-RAS and GAP through steric hindrance \cite{AS62} and impairs GAP-stimulated GTP hydrolysis. As a result, the mutant K-RAS (denoted by K-RAS\textsuperscript{G12C}) remains in the active state for a much longer period of time, leading to overactivation of the signaling pathway.\footnote{For the effects of mutations on the local electrostatic environment, see \cite{HM63}.} The mutation does not perturb RAF and GEF binding. 

AMG510\footnote{Also know as \textit{sotorasib}, molecular mass $561$ g/mol, half-life $5.5$ hours, binding energy ($\triangle{G}$) $-88$ kcal/mol) \cite{SL64, IS65}} and MRTX849\footnote{Also know as \textit{adagrasib}, molecular mass $604$ g/mol, half-life $23$ hours, binding energy ($\triangle{G}$) $-89$ kcal/mol) \cite{IS65, JR66}} are inhibitors of K-RAS\textsuperscript{G12C} that bind to an \textbf{regulatory site} (known as the switch-II pocket) of GDP-bound K-RAS\textsuperscript{G12C} and form covalent interactions with the cysteine-12 residue \cite{OP67, DB68}. 

Binding of these inhibitors (1) shifts the relative affinity of K-RAS\textsuperscript{G12C} in favor of GDP over GTP by partially disrupting the GDP/GTP binding site and (2) \textbf{allosterically impairs} RAF binding to K-RAS\textsuperscript{G12C}. The regulatory site is only accessible when K-RAS\textsuperscript{G12C} is bound to GDP, so the inhibitor must be bound over a significant number of K-RAS\textsuperscript{G12C} cycles to effectively block K-RAS\textsuperscript{G12C}. This is why a stable covalent bond with the cysteine-12 is needed.\footnote{Cysteine residues can be covalently bonded not only to each other via disulfide bonds (S-S bonds) but also to a variety of molecules.}

In the case of mutation of residue 12 from glycine (G) to aspartate (D), the mutant K-RAS\textsuperscript{G12D} lacks the reactive cysteine residue present in K-RAS\textsuperscript{G12C}, making it a major challenge to design selective compounds that bind to K-RAS\textsuperscript{G12D} in a stable manner \cite{DB68}. Using the valuable insights and inspiration provided by the development of K-RAS\textsuperscript{G12C} inhibitors, a non-covalent KRAS\textsuperscript{G12D} inhibitor called MRTX1133\footnote{Molecular mass $601$ g/mol, half-life $50$ hours, binding energy ($\triangle{G}$) $-73$ kcal/mol.} was developed in 2021 \cite{WW69, MI70, IM71}. 

MRTX1133 is a noncovalent inhibitor of KRAS\textsuperscript{G12D} that bind to the switch-II pocket of both GDP-bound and GDP-bound K-RAS\textsuperscript{G12D}s with high affinity without the requirement for covalent interactions. It works by \textbf{allosterically impairing} RAF binding to K-RAS\textsuperscript{G12D}. Note that the absolute value of the binding energy is comparable to that of the covalent inhibitors mentioned above. Because noncovalent interactions act cooperatively, the small energy contributions from each weak interaction add up to a significant amount of energy \cite{IM71, Cl35}.

\begin{figure}
\centering
\captionsetup{width=1.0\linewidth}
\includegraphics{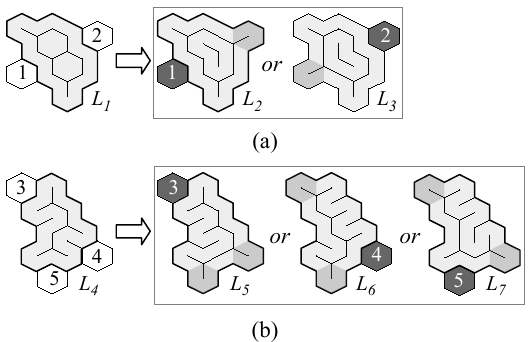}
\caption{(a) First-order inhibitors. White hexagons indicate h-active sites and black hexagons indicate non h-active sites. $Hex_1$ dose not bind to $L_2=L_1+Hex_2$ and $Hex_2$ dose not bind to $L_3=L_1+Hex_1$, (b) Second-order inhibitors. $Hex_3$ dose not bind to $L_5=(L_4+Hex_4)+Hex_5$, $Hex_4$ dose not bind to $L_6=(L_4+Hex_3)+Hex_5$, and $Hex_5$ dose not bind to $L_7=(L_4+Hex_3)+Hex_4$.}
\label{figure10}
\end{figure}

\subsubsection{Reproduction of allosteric regulation}\label{sec432}

Now that we have some understanding of intermolecular interactions, let's return to the model of intermolecular interactions proposed\footnote{See Definition \ref{M_of_II}}. To mimic the binding of small molecule drugs to a protein, we consider binding of hexagons (i.e., loops of length $6$) to a loop.

\begin{definition} [h-active sites]
Given two molecules, a loop $L_a$ and a hexagon $H_a$. Suppose that $L_a$ and $H_a$ interact. The \textit{h-active site} of $L_a$ for $H_a$ is the location of the boundary of $L_a$ where $H_a$ binds.
\end{definition}

The h-active sites of a molecule are uniquely determined by the shape of the molecule.

\begin{example} In Figure \ref{figure10} (a), both of site $Site_1$ (white hexagon numbered 1) and $Site_2$ (white hexagon numbered 2) of $L_1$ are h-active, $Site_1$ (black hexagon numbered 1) of $L_2$ is not h-active, and $Site_2$ (black hexagon numbered 2) of $L_3$ is not h-active.
\end{example}

\begin{definition} [$n$-th order inhibition/activation]
\textit{$n$-th order inhibition} or \textit{$n$-th order activation} is the inhibition or activation (of the binding of a molecules) by the binding of $n$ molecules, respectively. \textit{Inhibitors} are molecules that cause inhibition. \textit{Activators} are molecules that cause activation.
\end{definition}

\begin{example} In Figure \ref{figure10} (a), $Hex_1$ is a first order inhibitor of $Hex_2$ and vice versa, i.e., binding of $Hex_1$ to $L_1$ inhibits binding of $Hex_2$ and vice versa. In Figure \ref{figure10} (b), each pair of hexagons $Hex_i$ and $Hex_j$ forms a second-order inhibitor of $Hex_k$, where $\{i, j, k\}=\{3, 4, 5\}$.
\end{example}

\begin{remark} In the example above, AMG510 (or MRTX849) binds to K-RAS\textsuperscript{G12C} only if a GDP molecule binds to K-RAS\textsuperscript{G12C} (first-order activation). AMG510 (or MRTX849) dose not bind to K-RAS\textsuperscript{G12C} due to steric hindrance when a GTP molecule is bound to K-RAS\textsuperscript{G12C} (first-order inhibition). GDP and GTP molecules do not bind to K-RAS\textsuperscript{G12C} at the same time because they share a binding site (first-order inhibition). Neither GAP nor RAF binds to K-RAS\textsuperscript{G12C} when both AMG510 (or MRTX849) and GDP molecules bind to K-RAS\textsuperscript{G12C} (second-order inhibition).
\end{remark}

Using examples, let’s examine \textit{couplings} of bindings on the h-active sites of a molecule. For simplicity, we only consider first-order activations/inhibitions of h-active sites.

\begin{figure}
\centering
\captionsetup{width=1.0\linewidth}
\includegraphics{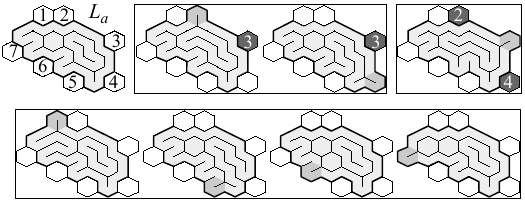}
\caption{Seven h-active sites $\{Site_1, Site_2, \cdots , Site_7\}$ of molecule $L_a$. White hexagons indicate h-active sites and black hexagons indicate non h-active sites. $Hex_k$ binds to $Site_k$, inhibiting the binding of the hexagons colored black. For example, $Hex_2$ binds to $Site_2$, inhibiting $Site_3$. The seven h-active sites are classified into three groups according to the pattern of inhibited h-active sites: $\{Site_2, Site_4\}$, $\{Site_3\}$, and $\{Site_1, Site_5, Site_6, Site_7\}$.}
\label{figure11}
\end{figure}

\begin{figure}
\centering
\captionsetup{width=1.0\linewidth}
\includegraphics{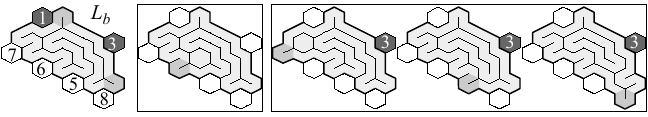}
\caption{Four h-active sites $\{Site_5, Site_6, Site_7, Site_8\}$ of molecule $L_b=(L_a+Hex_2)+Hex_4$, where $L_a$ is shown in Figure \ref{figure11}. White hexagons indicate h-active sites and black hexagons indicate non h-active sites. $Site_1$ and $Site_3$ of $L_b$ are inactive. (Note that the inhibition of $Site_1$ is a 2nd-order.) $Hex_k$ binds to $Site_k$, inhibiting the binding of the hexagons colored black. The four h-active sites are classified into two groups according to the pattern of h-active sites: $\{Site_6\}$ and $\{Site_5,Site_7,Site_8\}$.}
\label{figure12}
\end{figure}

\begin{definition} [Druggable h-active sites]
h-active sites are called \textit{druggable} if there is a first order activator (or inhibitor) which activates (or inhibits) the binding of hexagons to the h-active site, but does not activate binding to other active sites (i.e., no side-effects). 
\end{definition}

\begin{example} [Inhibitors] In Figure \ref{figure11}, the binding of $Hex_2$ at $Site_2$ (or $Hex_4$ at $Site_4$) inhibits $Site_3$. The binding of $Hex_3$ at Site3 inhibits two h-active sites $Site_3$ and $Site_5$. The binding of the other hexagons to $L_a$ inhibits no h-active sites. Thus, only $Site_3$ is druggable (by $Hex_2$ or $Hex_4$).
\end{example}

\begin{example} [Activators] In Figure \ref{figure12}, the binding of $Hex_6$ at $Site_6$ activates both $Site_1$ and $Site_3$. The binding of $Hex_5$ at $Site_5$ (or $Hex_7$ at $Site_7$ or $Hex_8$ at $Site_8$) activates $Site_1$. Thus, only $Site_1$ is druggable (by $Hex_5$, $Hex_7$, or $Hex_8$). 
\end{example}

\begin{remark} In Figure \ref{figure11}, the length ratio of $L_a$ ($88$ triangles) to a hexagon ($6$ triangles) is approximately $15$. On the other hand, the K-Ras protein (molecular mass $22$ kg/mol) interacts with a GTP molecule ($523$ g/mol), a GDP molecule ($443$ g/mol), a GAP protein ($120 \sim 130$ kg/mol), a GEF protein ($110 \sim 150$ kg/mol), a RAF protein ($70 \sim 75$ kg/mol), and small molecule inhibitors ($561 \sim 604$ g/mol). In particular, the mass ratio of K-RAS to inhibitor is approximately $40$. 
\end{remark}

\section{Discussion}\label{sec5}

A clear and simple mathematical model of intermolecular interactions is proposed to elucidate \textit{protein allosteric regulation}, the basis of protein communication and interaction. It is based on the concept of \textit{electron delocalization}, one of the main features of quantum chemistry, and provides new insights into \textit{coupling of bindings of molecules on a protein}. Allosteric regulation then follows tautologically from the definition of intermolecular interactions. However, since this is a toy model, it cannot predict actual protein behavior. The problem here is how to verify the validity of the model without verification of predictions. That is, dose this model provide a basis for protein communication and interaction?

So far, there's no reliable and universally accepted way to determine how accurately a model of protein interactions reflects the real-world situation it's supposed to represent. For example, in the orbital approximation of many-electron systems, the orbitals are fictitious and not physically observable. Interpretation based on orbitals is therefore very intuitive (see Section \ref{sec2}). Instead, we would like to examine the usefulness of the proposed model here. That is, whether this model will provide a new approach to elucidating the fundamental principles of protein communication and interaction or not.

First, this model is not only clear and simple, but also complex enough to reproduce allosteric regulation of proteins (Figure \ref{figure9}, \ref{figure11}, \ref{figure12}). We can then calculate immediately both (1) the location of active sites of a protein and (2) the presence or absence of allosteric regulation between them. On the other hand, in the previous models, both (1) and (2) are given in advance and neither is obtained by calculation. As a result, previous models are often example-specific and have not revealed the underlying mechanism of allostery without conformational changes. Thus, it can be said that the proposed model actually provides a new approach to elucidating the fundamental principles of protein communication and interaction. 

Features of the proposed model include:
\begin{enumerate}
\item	Geometric: Whether or not molecules interact depends solely on their shape.
\item	Global: Binding at an active site affects binding at a distal active site.
\item	No conformational change: Molecules are not deformed by intermolecular interactions. 
\item	Loop representation: The state of a protein corresponds to a decomposition of the protein into a loop complex. Loops correspond to the states of the electron cloud of a protein, not to the folding of the amino acid sequence of a protein.
\item	No memory (of reaction pathways): Given three molecules $L_1$, $L_2$, and $L_3$. Then, $(L_1+L_2)+L_3=(L_1+L_3)+L_2$ if both sides of the equation are well defined. Note that either $(L_1+L_2)+L_3$ or $(L_1+L_3)+L_2$ may not be well defined due to allosteric regulation.
\end{enumerate}

These are primarily direct consequences of the top-down description of molecules, which implicitly incorporates an effect of quantum mechanics, i.e., electron delocalization, into the model.\footnote{Recall that all intermolecular interactions are quantum mechanical in origin.} In contrast, no description of shape is available and the incorporation of quantum effects is not obvious in the conventional bottom-up representation of a protein, i.e., network of amino-acid residues.

Moreover, on the technical side, the top-down description allows us to apply global theories of mathematics to the problems of proteins. For example, when the region covered by a loop has holes inside, It is often possible to determine the existence of loop decompositions of the holes without filling them. If a loop decomposition exists for every hole, the loop and the holes constitute a loop decomposition of the entire region.

Limitations of the proposed model include:
\begin{enumerate}
\setcounter{enumi}{5}
\item Sharing of electrons only: Not considered explicitly are electrostatic interactions induced by dynamic electron correlation, the entropy-driven hydrophobic effect, and other types of non-covalent interactions. 
\item No distinction between covalent interactions and non-covalent interactions. As a result, the distinction between molecules and molecular complices is no longer self-evident.
\item No consideration of entropy: The contribution of entropy S (and temperature T) in Gibbs free energy $G=H-TS$ (Subsection \ref{sec41}) is not explicitly considered. (Entropy effects are implicitly included in the model because \textit{delocalization} of particles increases the \textit{uncertainty} of the system.) 
\item Two-dimensional: For simplicity, we only consider the case of $2$-simplices, i.e. triangles.
\item Flows on a flat mesh $M_0$ only. 
\end{enumerate}

As for (6), it can be said that the state of the electron cloud, which is the result of all interactions between nuclei and electrons of a molecule, is directly considered. However, it may be necessary to consider other types of non-covalent interactions as well when predicting the real world phenomena.

As for (7), this is a consequence of (6), and the clarity and simplicity of the model may owe much to (6) and (7).

As for (8), this is a consequence of the fact that the objective of this study is not the thermodynamical description of allostery \textit{phenomenon}, but the geometrical description of allostery \textit{mechanism}. 

As for (9), analysis of $3D$ models is one of the future challenges. However, there is a gap between the $2D$ model and the $3D$ model, and generalization to higher dimensions is not so obvious. For example, in the $3$-dimensional case, a unit cube $[0,1]^4$ is projected onto a rhombic dodecahedron, which consists of four tetrahedral loops of length six. More steps are thus required to find a loop of tetrahedra that sweeps a given $3$-dimensional region. 

As for (10), flows on a non-flat mesh can be obtained by considering flows of $2$-simplices induced on the surfaces of a trajectory of $n$-simplices ($n>2$). For more realistic modeling of protein interactions, studies on flows on a non-flat mesh may be required in the future.

Finally, this paper is also intended as a concise introduction to Quantum Chemistry, Protein Allosteric Regulation, Intermolecular Interactions, and Small Molecule Allosteric Drugs. I tried to explain the reasoning behind the theories in plain language as best I could. To my knowledge, there is no other literature that collectively explains these topics to non-specialists. I hope this paper will provide a starting point for many mathematicians to study chemistry and molecular biology.


\begin{thebibliography}{99}
\bibitem{PB1}
P. BALL, Beyond the bond, \textit{Nature}, \textbf{469} (2011), 26–28.
\bibitem{ZS2}
L. ZHAO, W. H. E. SCHWARZ and G. FRENKING, The Lewis electron-pair bonding model: the physical background, one century later, \textit{Nature Reviews Chemistry}, \textbf{3} (2019), 35–47.
\bibitem{CH3}
E. C. Constable and C. E. Housecroft, Chemical Bonding: The Journey from Miniature Hooks to Density Functional Theory, \textit{Molecules}, \textbf{25}(11) (2020), 2623.
\bibitem{CA4}
C. A. Coulson, \textit{The Spirit of Applied Mathematics: An Inaugural Lecture Delivered Before the University of Oxford on 28 October 1952}, Clarendon Press, Oxford, 1953.
\bibitem{RS5}
R.S. Mulliken, Molecular scientists and molecular science – some reminiscences, \textit{J. Chem. Phys.}, \textbf{43}(10) (1965), S2–S11.
\bibitem{NM6}
N. Morikawa, Global Geometrical Constraints on the Shape of Proteins and Their Influence on Allosteric Regulation. \textit{Applied Mathematics}, \textbf{9}(10) (2018), 1116-1155. 
\bibitem{MJ7}
J. Monod and F. Jacob, General Conclusions: Teleonomic Mechanisms in Cellular Metabolism, Growth, and Differentiation, \textit{Cold Spring Harb. Symp. Quant. Biol.}, \textbf{26} (1961), 389-401.
\bibitem{Ce8}
E. D. Cera, Mechanisms of ligand binding, \textit{Biophys. Rev. (Melville)}, \textbf{1}(1) 2020, 011303.
\bibitem{LN9}
J. Liu and R. Nussinov, Allostery: An Overview of Its History, Concepts, Methods, and Applications, \textit{PLoS Comput. Biol.}, \textbf{12}(6) (2016), e1004966.
\bibitem{WM10}
S. J. Wodak, et al., Allostery in Its Many Disguises: From Theory to Applications (2019), \textit{Structure}, \textbf{27}(4) (2019), 566-578.
\bibitem{MW11}
H. N. Motlagh, James O. Wrabl, J, Li and V. J. Hilser, The ensemble nature of allostery, \textit{Nature}, \textbf{508} (2014), 331–339.
\bibitem{MC12}
M. Montserrat-Canals, G. Cordara and U. Krengel, Allostery, \textit{Q. Rev. Biophys.}, \textbf{58} (2025), e5.
\bibitem{GT13}
R. G. Govindaraj, S. Thangapandian, M. Schauperl, R. A. Denny and D. J. Diller, Recent applications of computational methods to allosteric drug discovery, \textit{Front. Mol. Biosci.}, \textbf{9} (2023), 1070328.
\bibitem{CM14}
M. Civera, E. Moroni, L. Sorrentino, F. Vasile and S. Sattin, Chemical and Biophysical Approaches to Allosteric Modulation, \textit{Eur. J. of Org. Chem.}, \textbf{2021}(30) (2021), 4245-4259.
\bibitem{MW15}
M. Williamson, \textit{How Proteins work}, Garland Science, New York, 2012.
\bibitem{SE16}
R. D Smith, A. L. Engdahl, J. B. Dunbar Jr. and H. A. Carlson, Biophysical limits of protein-ligand binding, \textit{J. Chem. Inf. Model.}, \textbf{52}(8) (2012), 2098-2106.
\bibitem{SI17}
M. W. Schmidt, J. Ivanic and K. Ruedenberg, The physical origin of covalent binding, in \textit{The Chemical Bond: Fundamental Aspects of Chemical Bonding}, ed. G. Frenking and S. Shaik, 1–67, Wiley-VCH, Weinheim, 2014.
\bibitem{LG18}
D. S. Levine and M. Head-Gordon, Clarifying the quantum mechanical origin of the covalent chemical bond, \textit{Nat. Commun.}, \textbf{11}(1) (2020), 4893.
\bibitem{F19}
G. Frenking, The Chemical Bond – an Entrance Door of Chemistry to the Neighboring Sciences and to Philosophy, \textit{Isr. J. of Chem.}, \textbf{62}(36) (2021), 1-9.
\bibitem{PC20}
E. Pastorczak and C. Corminboeuf, Perspective: Found in translation: Quantum chemical tools for grasping non-covalent interactions, \textit{J. Chem. Phys.}, \textbf{146}(12) (2017), 120901.
\bibitem{CW20a} 
K. Crane and M. Wardetzky, A Glimpse into Discrete Differential Geometry, \textit{Notices of the A.M.S.}, \textbf{64} (2017), 1153-1159.
\bibitem{GDS20b} 
P. Schröder, E. Grinspun and M. Desbrun, \textit{Discrete differential geometry: an applied introduction}, ACM SIGGRAPH'05 Course Notes, (2005).
\bibitem{BS20c} A. Bobenko and Y. Suris, \textit{Discrete Differential Geometry: Integrable Structure}, Grad. Stud. Math., \textbf{98}, A.M.S., Providence, 2008.
\bibitem{NM20d} L.Mavridis, et al, \textit{SHREC’10 Track: Protein Models, Eurographics Workshop on 3D Object Retrieval}, ed. I. Pratikakis, M. Spagnuolo, T. Theoharis, and R. Veltkamp, 2010.
\bibitem{NM20e} N.Morikawa, On the Defining Equations of Protein's Shape from a Category Theoretical Point of View, \textit{Applied Mathematics}, \textbf{11}(9) (2020), 907-946.
\bibitem{LM21}
G. Li, D. Magana and R. B. Dyer, Anisotropic energy flow and allosteric ligand binding in albumin, \textit{Nat. Commun.}, \textbf{5} (2014), 3100.
\bibitem{CD22}
A. Cooper and D. T. F. Dryden, Allostery without conformational change. A plausible model, \textit{Eur. Biophys. J.}, \textbf{11}(2) (1984), 103-109.
\bibitem{TK23}
S.-R. Tzeng and C. G. Kalodimos, Protein activity regulation by conformational entropy, \textit{Nature}, \textbf{488} (2012), 236–240.
\bibitem{W24}
A. J. Wand, Deep mining of the protein energy landscape, \textit{Struct. Dyn.}, \textbf{10}(2) (2023), 020901.
\bibitem{HR25}
G. Haran, I. Riven, Perspective: How Fast Dynamics Affect Slow Function in Protein Machines, \textit{J. Phys. Chem. B}, \textbf{127}(21) (2023), 4687-4693.
\bibitem{GZ26}
J. Guo and H.-X. Zhou, Protein Allostery and Conformational Dynamics, \textit{Chem. Rev.}, 116(11) (2016), 6503-6515.
\bibitem{NV27}
C. F. A. Negre et al,, Eigenvector centrality for characterization of protein allosteric pathways, \textit{Proc. Natl. Acad. Sci. U.S.A.}, \textbf{115}(52) (2018), E12201–E12208.
\bibitem{VC28}
Y. L. Vishweshwaraiah, J. Chen and N. V. Dokholyan, Engineering an Allosteric Control of Protein Function, \textit{J. Phys. Chem. B}, \textbf{125}(7) (2021), 1806-1814.
\bibitem{MW29}
L. K. Madan, C. L. Welsh, A. P. Kornev and S. S. Taylor, The “violin model”: Looking at community networks for dynamic allostery, \textit{J. Chem. Phys.}, 158(8) (2023), 081001.
\bibitem{CT30}
D. Merkle et al., \textit{From Category Theory to Enzyme Design: Unleashing the Potential of Computational Systems Chemistry}, Algorithmic Cheminformatics Group, Univ. Southern Denmark, 2020, \url{https://cheminf.imada.sdu.dk/novo-synergy}
\bibitem{Be31}
A. D. Becke, Perspective: Fifty years of density-functional theory in chemical physics, \textit{J. Chem. Phys.}, \textbf{140}(18) (2014), 18A301.
\bibitem{AS32}
V. A. Adhav and K. Saikrishnan, The Realm of Unconventional Noncovalent Interactions in Proteins: Their Significance in Structure and Function, \textit{ACS Omega}, \textbf{8}(25) (2023), 22268-22284.
\bibitem{NR33}
R. W. Newberry and R. T. Raines, Secondary Forces in Protein Folding, \textit{ACS Chem. Biol.}, \textbf{14}(8) (2019), 1677-1686.
\bibitem{CM34}
S. A. Combs, B. K. Mueller and J. Meiler, Holistic Approach to Partial Covalent Interactions in Protein Structure Prediction and Design with Rosetta, \textit{J. Chem. Inf. Model.}, \textbf{58}(5) (2018), 1021-1036.
\bibitem{Cl35}
T. Clark, How deeply should we analyze non-covalent interactions?, \textit{J. Mol. Model.}, \textbf{29}(3) (2023), 66.
\bibitem{We36}
F. Weinhold, “Noncovalent Interaction”: A Chemical Misnomer That Inhibits Proper Understanding of Hydrogen Bonding, Rotation Barriers, and Other Topics, \textit{Molecules}, \textbf{28}(9), (2023), 3776.
\bibitem{VG37}
F. Vascon et at., Protein electrostatics: From computational and structural analysis to discovery of functional fingerprints and biotechnological design, \textit{Comput. Struct. Biotechnol. J.}, \textbf{18} (2020), 1774-1789.
\bibitem{ZP38}
H.-X. Zhou and X. Pang, Electrostatic Interactions in Protein Structure, Folding, Binding, and Condensation, \textit{Chem Rev.}, \textbf{118}(4) (2018), 1691–1741.
\bibitem{He39}
J. M. Herbert, Neat, Simple, and Wrong: Debunking Electrostatic Fallacies Regarding Noncovalent Interactions, \textit{J. Phys. Chem. A}, \textbf{125}(33) (2021), 7125-7137.
\bibitem{PF40}
C. N. Pace et al., Contribution of hydrophobic interactions to protein stability, \textit{J. Mol. Biol.}, \textbf{408}(3) (2011), 514-528.
\bibitem{WK41}
F. Weinhold and R. A. Klein, What is a hydrogen bond? Resonance covalency in the supramolecular domain, \textit{Chem. Educ. Res. Pract.}, \textbf{15}(3) (2014), 276-285.
\bibitem{GB42} \textit{'hydrogen bond' in IUPAC Compendium of Chemical Terminology, 5th ed.}, International Union of Pure and Applied Chemistry; 2025. Online version 5.0.0, 2025.
\bibitem{Fe43} R. P. Feynman, Forces in Molecules, \textit{Phys. Rev.}, \textbf{56} (1939), 340-343.
\bibitem{HD44} J. Hermann, R. A. DiStasio Jr. and A. Tkatchenko, First-Principles Models for van der Waals Interactions in Molecules and Materials: Concepts, Theory, and Applications, \textit{Chem. Rev.}, \textbf{117}(6) (2017), 4714-4758.
\bibitem{Hu45} K. L. C. Hunt, \textit{A chemist’s perspective on van der Waals dispersion forces: Challenges and opportunities}, QuFiCh Workshop, October 10, 2023.
\bibitem{Ch46} L.-Q. Chen, Chemical potential and Gibbs free energy, \textit{MRS Bulletin}, \textbf{44} (2019), 520 – 523.
\bibitem{ZR47} R. K. P. Zia, E. F. Redish and S. R. McKay, Making sense of the Legendre transform, \textit{Am. J. Phys.}, \textbf{77}(7) (2009), 614-622.
\bibitem{Le48} H. S. Leff, Removing the Mystery of Entropy and Thermodynamics – Part V, \textit{Phys. Teach.}, \textbf{50}(5) (2012), 274–276.
\bibitem{Ro49} O. Rodriguez, \textit{CHEM1306: Health Chemistry I (Rodriguez), 7.5: Bond Energies}, El Paso Community College, 2023, \url{https://chem.libretexts.org/Courses/El_Paso_Community_College}
\bibitem{WK50} \textit{Non-covalent interaction}. (n.d.). In Wikipedia. Retrieved from \url{https://en.wikipedia.org/wiki/Non-covalent_interaction}
\bibitem{RK51} K. Roy, S. Kar and R. N. Das,, \textit{Understanding the Basics of QSAR for Applications in Pharmaceutical Sciences and Risk Assessment}, Academic Press, an imprint of Elsevier, Amsterdam, 2015.
\bibitem{KC52} I. D. Kuntz, K. Chen, K. A. Sharp and P. A. Kollman, The maximal affinity of ligands, \textit{Proc. Natl. Acad. Sci. U.S.A.}, \textbf{96}(18) (1999), 9997–10002.
\bibitem{Co53} J. Corzo, Time, the forgotten dimension of ligand binding teaching, \textit{Biochem. Mol. Biol. Educ.}, \textbf{34}(6) (2006), 413-416.
\bibitem{WK54} \textit{Equipartition theorem}. (n.d.). In Wikipedia. Retrieved from \url{https://en.wikipedia.org/wiki/Equipartition_theorem}
\bibitem{LK55} Q. Li and C.-B. Kang, Mechanisms of Action for Small Molecules Revealed by Structural Biology in Drug Discovery, \textit{Int. J. Mol. Sci.}, \textbf{21}(15) (2020), 5262.
\bibitem{XB56} X. Barril, \textit{Allosteric drugs: A differentiated small molecule approach}, Drug Discovery \& Development, 2023, \url{https://www.drugdiscoverytrends.com/allosteric-drugs-a-differentiated-small-molecule-approach}
\bibitem{GT57} R. G. Govindaraj, S. Thangapandian, M. Schauperl, R. A. Denny and D. J. Diller, Recent applications of computational methods to allosteric drug discovery, \textit{Front. Mol. Biosci.}, \textbf{9} (2023), 1070328.
\bibitem{CF58} M. A. Coban, S. Fraga, T. R. Caulfield, Structural And Computational Perspectives of Selectively Targeting Mutant Proteins, \textit{Curr. Drug. Discov. Technol.}, \textbf{18}(3) (2021), 365-378.
\bibitem{HM59} H. McGavock, \textit{How Drugs Work  Basic Pharmacology for Health Professionals, Fourth Edition}, McGavock, London, 2016.
\bibitem{DC60} M. J. Duffy and J. Crown, Drugging "undruggable" genes for cancer treatment: Are we making progress?, \textit{Int. J. Cancer}, \textbf{148}(1) (2021), 8-17.
\bibitem{CO61} J G Christensen 1, P Olson 1, T Briere 1, C Wiel 2, M O Bergo, Targeting Krasg12c -mutant cancer with a mutation-specific inhibitor, \textit{J. Intern. Med.}, \textbf{288}(2) (2020), 183-191.
\bibitem{AS62} M. R. Ahmadian, P. Stege, K. Scheffzek and A. Wittinghofer, Confirmation of the arginine-finger hypothesis for the GAP-stimulated GTP-hydrolysis reaction of Ras, \textit{Nat. Struct. Biol.}, \textbf{4}(9) (1997), 686-689.
\bibitem{HM63} J. C. Hunter, A. Manandhar, M. A. Carrasco, D. Gurbani, S. Gondi and K. D. Westover, Biochemical and structural analysis of common cancer-associated KRAS mutations, \textit{Mol. Cancer Res.}, \textbf{13}(9) (2015), 1325-35.
\bibitem{SL64} F. Skoulidis et al., Sotorasib for Lung Cancers with KRAS p.G12C Mutation, \textit{N. Engl. J. Med.}, \textbf{384}(25) (2021), 2371-2381. 
\bibitem{IS65} A. R. Issahaku, E. Y. Salifu and M. E. S. Soliman, Inside the cracked kernel: establishing the molecular basis of AMG510 and MRTX849 in destabilising KRASG12C mutant switch I and II in cancer treatment, \textit{J. Biomol. Struct. Dyn.}, \textbf{41}(11) (2023), 4890-4902.
\bibitem{JR66} P. A. Jänne et al., Adagrasib in Non-Small-Cell Lung Cancer Harboring a KRASG12C Mutation, \textit{N. Engl. J. Med.}, \textbf{387}(2) (2022), 120-131.
\bibitem{OP67} J. M. Ostrem, U. Peters, M. L. Sos, J. A. Wells and K. M. Shokat, K-Ras(G12C) inhibitors allosterically control GTP affinity and effector interactions, \textit{Nature}, \textbf{503}(7477) (2013), 548-51.
\bibitem{DB68} M. Drosten and M. Barbacid, KRAS inhibitors: going noncovalent, \textit{Mol. Oncol.}, \textbf{16}(22) (2022), 3911-3915.
\bibitem{WW69} D. Wei, L. Wang, X. Zuo, A. Maitra and R. S. Bresalier, A small molecule with big impact: MRTX1133 targets the KRASG12D mutation in pancreatic cancer, \textit{Clin. Cancer. Res.}, \textbf{30}(4) (2024), 655–662.
\bibitem{MI70} X. Wang et al., Identification of MRTX1133, a Noncovalent, Potent, and Selective KRASG12D Inhibitor, \textit{J. Med. Chem.}, \textbf{65}(4) (2022), 3123-3133.
\bibitem{IM71} A. R. Issahaku et al., Characterization of the binding of MRTX1133 as an avenue for the discovery of potential KRASG12D inhibitors for cancer therapy, \textit{Sci. Rep.}, \textbf{12}(1) (2022), 17796.
\end{thebibliography}
\end{document}